\documentclass[11pt,letter]{article}
\pdfoutput=1
\usepackage{wrapfig}
\usepackage{epstopdf}
\usepackage{amssymb,latexsym}
\usepackage{graphicx,amsmath,amsfonts}
\usepackage{amsthm}
\usepackage{color}
\usepackage{subfig}
\usepackage{comment}

\begin{document}

\newtheorem{theorem}{Theorem}[section]
\newtheorem{corollary}[theorem]{Corollary}
\newtheorem{lemma}[theorem]{Lemma}
\newtheorem{remark}[theorem]{Remark}
\newtheorem{definition}[theorem]{Definition}
\newtheorem{example}[theorem]{Example}
\theoremstyle{plain}
\newtheorem{proposition}[theorem]{Proposition}
\newtheorem{observation}[theorem]{Observation}
\newcommand{\NN}{\mbox{$\cal N$}}
\newcommand{\PP}{\mbox{$\cal P$}}
\newcommand{\WW}{\mbox{$\cal W$}}
\newcommand{\RR}{\mbox{$\cal R$}}
\newcommand{\A}{\mbox{$\cal A$}}
\newcommand{\calS}{\mbox{$\cal S$}}
\newcommand{\poly}{\mathrm{poly}}

\def\R{\mathbb{R}}
\def\Z{\mathbb{Z}}
\def\N{\mathbb{N}}
\def\ds{\displaystyle}
\def\oa#1{\overrightarrow{#1}}
\def\ola#1{\overleftarrow{#1}}
\def\row#1#2{{#1}_1,\ldots ,{#1}_{#2}}
\def\brow#1#2{{\bf {#1}}_1,\ldots ,{\bf {#1}}_{#2}}
\def\rrow#1#2{{#1}_0,{#1}_1,\ldots ,{#1}_{#2}}
\def\urow#1#2{{#1}^1,\ldots ,{#1}^{#2}}
\def\irow#1#2{{#1}_1,\ldots ,{#1}_{#2},\ldots}
\def\lcomb#1#2#3{{#1}_1{#2}_1+{#1}_2{#2}_2+\cdots +{#1}_{#3}{#2}_{#3}}
\def\blcomb#1#2#3{{#1}_1{\bf {#2}}_1+\cdots +{#1}_{#3}{\bf {#2}}_{#3}}
\def\2vec#1#2{\left(\begin{array}{c}{#1}\\{#2}\end{array}\right)}
\def\threevec#1#2#3{\left[\begin{array}{r}{#1}\\{#2}\\{#3}\end{array}\right]}
\def\mod#1{\ \hbox{\rm (mod $#1$)}}
\def\gcd#1#2{\hbox{gcd}\>(#1,#2)}
\def\lcm#1#2{\hbox{lcm}\>(#1,#2)}
\def\card#1{\hbox{card}\>(#1)}
\def\adv{\hbox{adv}}
\def\union{\cup}
\def\intsct{\cap}
\def\Union{\bigcup}
\def\Intsct{\bigcap}

\newcommand{\calA}{\mathcal{A}}
\newcommand{\calE}{\mathcal{E}}
\newcommand{\calR}{\mathcal{R}}
\newcommand{\calF}{\mathcal{F}}
\newcommand{\calB}{\mathcal{B}}
\newcommand{\calC}{\mathcal{C}}
\newcommand{\calQ}{\mathcal{Q}}
\newcommand{\calD}{\mathcal{D}}
\newcommand{\calP}{\mathcal{P}}
\newcommand{\calT}{\mathcal{T}}
\newcommand{\calL}{\mathcal{L}}
\newcommand{\Dec}{\mathit{Dec}}
\newcommand{\revnot}[1]{\overleftarrow{#1}}

\newcommand{\peak}{{\mathrm{peak}}}
\newcommand{\dfs}{{\mathrm{dfs}}}

\newcommand{\p}{{{\mathrm{P}}}}
\newcommand{\np}{{{\mathrm{NP}}}}

\title{Clone Structures in Voters' Preferences}

\author{Edith Elkind \\
       School of Physical and \\Mathematical Sciences \\
       Nanyang Technological University\\
       Singapore
\and
       Piotr Faliszewski\\
       Department of Computer Science\\ 
       AGH University of Science and\\ Technology,
       Poland
\and
       Arkadii Slinko\\
       Deptartment of Mathematics\\
       University of Auckland\\
       New Zealand
}

\maketitle

\begin{abstract}
  In elections, a set of candidates ranked consecutively
  (though possibly in different order) by all voters is called a clone set, 
  and its members are called clones. A clone structure is a family of all
  clone sets of a given election. In this paper we study properties of
  clone structures. In particular, we give an axiomatic
  characterization of clone structures, show their hierarchical
  structure, and analyze clone structures in single-peaked and
  single-crossing elections. We give a polynomial-time algorithm that
  finds a minimal collection of clones that need to be collapsed for an
  election to become single-peaked, and we show that this problem
  is $\np$-hard for single-crossing elections.
\end{abstract}

\section{Introduction}\label{sec:intro}
Group decision making plays an important role in the proper
functioning of human societies and multiagent systems. Collective
decisions are often made by aggregating the preferences of individual
agents by means of {\em voting}: each agent ranks the available
alternatives, and a voting rule is used to select one or more winners
(see~\cite{arr-sen-suz:b:handbook-of-social-choice} for a general
overview of voting, and~\cite{fal-hem-hem:j:cacm-survey} for a more
algorithmic perspective).  In general, the structure of the set of
alternatives may be quite complex. For instance, Ephrati and
Rosenschein~\cite{eph-ros:j:multiagent-planning} explore the situation
where multiple agents try to coordinate their actions in order to
devise a global plan. There the space of alternatives, i.e., of
possible plans, may be huge, with some alternatives being very similar
to each other.  In such a case it may be reasonable to establish which
plans differ fundamentally, and which are viewed as minor variations
of each other.

Such structured decision-making environments have been studied in the
social choice literature: for instance, Laffond et
al.~\cite{laf-lai-las:j:composition} describe the situation when a
group of agents has to choose from a set that is partitioned into
several ``projects,'' where each project is defined as a set of
possible variants. In this setting, all agents are likely to rank the
variants of each projects contiguously.  This model was further
investigated by Laslier~\cite{las:j:rank-based,las:j:aggregation}.
Tideman~\cite{tid:j:clones} suggests a
different explanation of why several alternatives in an election may
be very similar to each other: a malicious party may try to
``duplicate'' an existing candidate in order to change the voting
outcome. This procedure is known as {\em cloning} and the alternatives
that appear together in all preference orders (though not
necessarily in the same order) are called {\em clones}. 
Elkind et al.~\cite{elk-fal-sli:c:cloning} study algorithms for
cloning and show that optimal cloning is easy for many voting
rules.

Both when clones arise naturally and when they are created by a
manipulator, it may be useful to understand the internal structure of
the resulting clone sets. Indeed, such an understanding could be
instrumental in uncovering hidden properties of voters' preferences
such as, for example, a hierarchical structure of the alternative set, 
or the fact that after collapsing a small number of clones the
election becomes single-peaked or single-crossing (informally, 
both single-peaked~\cite{bla:b:polsci:committees-elections} 
and single-crossing~\cite{mir:j:single-crossing} elections model societies
focused on a single issue, such as, e.g., taxation level).
In either case we could run the election in a better way by using
a more suitable voting rule: in the former case we can use
hierarchic voting, and in the latter case we can use the median
voter rule---which is known to be strategy-proof for single-peaked
(single-crossing) profiles---to select a group of clones, and then
pick the final winner among them. Such an approach is likely to
produce a better voting outcome as well as reduce the voters'
incentives for manipulation. 

Our goal in this paper is to provide a formal understanding of what
families of clone sets---which we call \emph{clone structures}---can
arise in elections (we give an axiomatic characterization), to provide
convenient means of representing them (we show that PQ-trees of Booth
and Lueker~\cite{boo-lue:j:consecutive-ones-property} very
conveniently describe clone structures), and to find a polynomial-time
algorithms that restores single-peakedness/single-crossingness in
elections by collapsing a minimal number of clones (we succeed for the
case of single-peakedness and prove $\np$-hardness for the case of
single-crossing).  We believe that our results are useful for
understanding the impact of clones in decision-making scenarios, and
will help in developing algorithms for settings where some of the
candidates may be very similar to each other. Due to space limit,
almost all proofs, as well as some discussions, are relegated to the
appendix.

\section{Preliminaries}\label{sec:prelim}
Given a finite set $C$ of {\em candidates} (or {\em alternatives}), a
{\em preference order} (or {\em ranking}) over $C$ is a total order over $C$,
i.e., a complete, transitive and antisymmetric relation on $C$.
Intuitively, a preference order is a ranking of the candidates from
the most desirable one to the least desirable one. 
By $\ola{\succ}$ we denote an order obtained by reversing order $\succ$,
that is, $j \ola\succ i$ if and only if $i\succ j$.
For two disjoint sets
$X, Y\subseteq C$ and an order $\succ$, we write $X\succ Y$ if
$x\succ y$ for all $x\in X$ and all $y\in Y$.  
Given two sets $X, Y\subseteq C$, we say that $X$ is a {\em proper
  subset} of $Y$ if $X\subseteq Y$ and $1<|X|< |Y|$. We say that $X$
and $Y$ {\em intersect non-trivially} and write $X\bowtie Y$ if $X\cap
Y\neq\emptyset$, $X\setminus Y\neq\emptyset$ and $Y\setminus
X\neq\emptyset$.

A {\em preference profile} $\calR=(\row Rn)$ on $C$ is a collection of
$n$ preference orders over $C$, where each order $R_i$, $1 \leq i \leq n$,
represents the preferences of the $i$-th voter; for readability, we
sometimes write $\succ_i$ in place of $R_i$.  An {\em election}
over $C$ is a pair $\calE = (C, \calR)$, where $\calR$ is a preference
profile over $C$.  A {\em voting rule} is a mapping $\calF$ that,
given an election $\calE$ over $C$, outputs a set
$\calF(\calE)\subseteq C$; the elements of $\calF(\calE)$ are called
the {\em election winners}. Many voting rules are used in practice and studied    
theoretically; see~\cite{arr-sen-suz:b:handbook-of-social-choice}.
However, since we focus on the nature of preference profiles, our
results do not depend on the choice of a voting rule.

\begin{example}\label{exmpl:rules}
  $\calR = (R_1, R_2, R_3)$
  with $R_1: a \succ_1 b \succ_1 c \succ_1 d$, $R_2: b \succ_2 d
  \succ_2 c \succ_2 a$, $R_3: a \succ_3 b \succ_3 d \succ_3 c$
  is a preference profile over $C = \{a,b,c,d\}$.
\end{example}

\noindent
The following
definition, inspired by~\cite{tid:j:clones}, is fundamental for our
work.

\begin{definition}
  Let $\calR=(\row Rn)$ be a preference profile over a candidate set
  $C$.  We say that a non-empty subset $X\subseteq C$ is a {\em clone
    set} for $\calR$ if
  for every $c,c'\in X$, every $a\in C\setminus X$, and every
  $i=1,2,\ldots, n$
  it holds that
  $c\succ_i a$ if and only if $c'\succ_i a$.
\end{definition}
Unlike Tideman~\cite{tid:j:clones}, we define singletons to be clone
sets;
in the election from Example~\ref{exmpl:rules} each of
$\{a\}$, $\{b\}$, $\{c\}$, $\{d\}$, $\{d,c\}$, $\{b,c,d\}$, and
$\{a,b,c,d\}$ is a~clone~set.

\section{Axiomatic Characterization of Clone Structures}\label{sec:char}
Our first goal is to understand which set families can be obtained
as clone structures.  That is, given a collection $\calC$ of subsets
of a candidate set $C$, we %
want to determine if there exists a preference profile $\calR$ over
$C$ such that each clone in $\calR$ appears in our collection and vice
versa; we will say that such $\calR$ {\em implements} $\calC$.  The
main technical results of this section are (a) an axiomatic
characterization of implementable collections of subsets, and (b) a
polynomial-time algorithm for recognizing such families. 

In this section, we will consider elections over the set
$[m]=\{1, \dots, m\}$.  We will write $[j, k]$ to denote $\{j, j+1,
\dots, k\}$ for $j, k\in[m]$.

\begin{definition}\label{def:c-structure}
  Given a profile $\calR=(\row Rn)$ over $[m]$, let
  ${\calC}(\calR)\subseteq 2^{[m]}$ be the collection of all clone
  sets for $\calR$.  We say that a family ${\calC}\subseteq 2^{[m]}$
  is a {\em clone structure} on $[m]$ if it is equal to
  ${\calC}(\calR)$ for some profile $\calR$ on $[m]$.
\end{definition}

We remark that clone structures are very close in spirit to clans 
in $2$-structures~\cite{ehr-har-roz:b:2-structures},
and many results in this section resemble those for clans; the
characterization given in this section and high-level proof approach are close in spirit to
those of M\"{o}hring~\cite{moe:j:decomposition}.  We will, however, present
a direct argument rather than translate these prior results, both because we
need intermediate results for the analysis of single-peaked elections
and because such a translation
is non-trivial and would obscure useful intuition.

\begin{example}
\label{exm:string}
Let $\calR$ consist of a single linear order $R: 1 \succ 2 \succ
\cdots \succ m$.  Then ${\calC}(\calR)=\{[i,j]\mid i\le j\}$ (see
Figure~\ref{fig:string}(a)).
Let $\calR'$ be a cyclic profile on $[m]$, i.e., 
$\calR'=(R_1, \dots, R_m)$, and the preferences 
of the $i$-th voter are given by
$R_i: i \succ_i i+1 \succ_i \cdots \succ_i m\succ_i 1\succ_i \cdots \succ_i i-1.$
Then ${\calC}(\calR')=\{[m]\}\cup \{\{i\}\mid i\in [m]\}$
(see Figure~\ref{fig:string}(b)).
\end{example}

\begin{wrapfigure}{r}{4cm}
  \vspace{-0.4cm}
  \begin{center}
  \subfloat[{A string of sausages.}]{%
  \resizebox{4cm}{!}{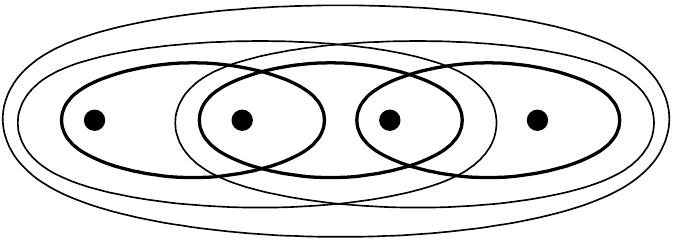}}
\vspace{0.2cm}
\subfloat[{A fat sausage.}]{%
  \raisebox{0.20cm}[\totalheight][0pt]{
  \resizebox{3.5cm}{!}{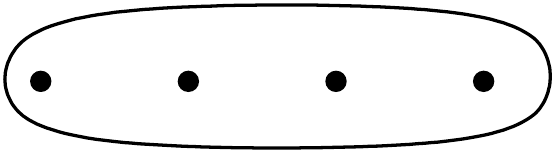}}}
  \end{center}
\vspace{-0.4cm}
\caption{\label{fig:string}Diagrams representing clone structures
    from Example~\ref{exm:string} for ${m = 4}$. }
  \vspace{-0.5cm}
\end{wrapfigure}

We call the first clone structure from Example~\ref{exm:string} a {\em
  string of sausages} and the second one a {\em fat sausage}.  Note
that any clone structure over $[m]$ consists of at most
$\frac{m(m+1)}{2}$ sets, since each clone set can be described by its
location (i.e., beginning and end) in the preference ordering of a
fixed voter.  Thus, a string of sausages and a fat sausage can be
thought of as, respectively, the maximal and the minimal clone
structure over $[m]$.
Let us now establish some basic properties of clone structures.

\begin{proposition}
\label{pro:inclusive}
Let $\calR$ be a profile on $[m]$. Then
(1) $\{i\} \in {\calC}(\calR)$ for any $i\in [m]$;
(2) $\emptyset \notin \calC(\calR)$ and $[m] \in {\calC}(\calR)$;
(3) if $C_1$, $C_2$ are in ${\calC}(\calR)$ and $C_1\cap
  C_2\ne \emptyset$, then $C_1\cup C_2$ and $C_1\cap C_2$ are in
  ${\calC}(\calR)$;
(4) if $C_1$, $C_2$ are in $\calC(\calR)$ and $C_1 \bowtie
  C_2$, then $C_1 \setminus C_2$ and $C_2 \setminus C_1$ are in
  $\calC(\calR)$.
\end{proposition}

Proposition~\ref{pro:inclusive} does not give sufficient conditions
for a family of subsets of $[m]$ to be a clone structure.  For
example, $P = 2^{[m]} \setminus \{\emptyset\}$, where $m \geq 3$,
satisfies all the conditions of Proposition~\ref{pro:inclusive}. Yet,
the cardinality of $P$ is $2^m-1$, whereas each clone structure over
$[m]$ has at most $\frac{m(m+1)}{2}$ elements.
The next proposition provides a further
necessary condition for a family of subsets of $[m]$ to be a clone
structure.  It is strong enough to exclude the collection $2^{[m]}
\setminus \{\emptyset\}$ for $m>3$.

Given a profile $\calR$ over $[m]$ and a set $X \in \calC(\calR)$, we
say that a set $Z \in \calC(\calR)$ is a {\em proper minimal superset}
of $X$ if $X \subseteq Z$, $X \neq Z$, and there is no set $Y \in
\calC(\calR)$ such that $X \neq Y$, $Y \neq Z$ and $X \subseteq Y
\subseteq Z$.

\begin{proposition}
\label{pro:exclusive1}
For any profile $\calR$ on $[m]$, each 
$X \in \calC(\calR)$ 
has at most two proper minimal supersets in $\calC(\calR)$.
\end{proposition}

Note, however, that for $m=3$ the set family
$2^{[m]}\setminus\{\emptyset\}$ satisfies the conclusion of
Proposition~\ref{pro:exclusive1}. Yet, it is obviously not a clone
structure, since it contains a ``cycle'' $\{1, 2\}, \{2, 3\}$, $\{3,
1\}$. More generally, consider a set family over $[m]$ that can be
obtained from a string of sausages by adding the ``missing link'',
i.e., the set $\{m, 1\}$ as well as all of its supersets that are
necessary to satisfy the conclusions of
Proposition~\ref{pro:inclusive}; we will call this set family a {\em
  ring of sausages}. Clearly, a ring of sausages is not a clone
structure, because it cannot be implemented by an acyclic preference
relation; yet, the conclusion of Proposition~\ref{pro:exclusive1} is
satisfied. Thus, we need to forbid rings of sausages;
in fact, we require a somewhat more general condition.

\begin{definition}
  We say that a set family $\{A_0, \dots, A_{k-1}\}$ is a {\em bicycle
    chain} if $k\ge 3$ and for all $i=0, \dots, k-1$ it holds that (1)
  $A_{i-1}\bowtie A_i$; (2) $A_{i-1}\cap A_i\cap A_{i+1}=\emptyset$;
  (3) $A_i\subseteq A_{i-1}\cup A_{i+1}$, where all indices are
  computed modulo $k$.
\end{definition}

\begin{proposition}~\label{pro:exclusive2} If $\calC$ is a
  clone structure, %
  it does not contain a bicycle chain.
\end{proposition}
Propositions~\ref{pro:inclusive},~\ref{pro:exclusive1}
and~\ref{pro:exclusive2} lead to the following set of axioms (note that
these axioms are not normative; they simply tell us what clone structures
\emph{are} and not what they \emph{should} be):

\begin{description}%
\vspace{-0.2cm}
\item[A1.] $\{f\} \in \calF$ for any $f\in F$, $\emptyset \notin \calF$, and $F \in \calF$.
\item[A2.] if $C_1$ and $C_2$ are in $\calF$ and $C_1\cap C_2\ne
  \emptyset$, then $C_1\cup C_2$ and $C_1\cap C_2$ are in
  $\calF$.
\item[A3.] If $C_1$ and $C_2$ are in $\calF$ and $C_1\bowtie C_2$,
  then $C_1 \setminus C_2$ and $C_2 \setminus C_1$ are  in
  $\calF$.
\item[A4.] Each $C \in \calF$ has at most two proper minimal supersets in $\calF$.
\item[A5.] $\calF$ does not contain a bicycle chain.
\vspace{-0.2cm}
\end{description}

Our next goal is to show that these five axioms indeed characterize
clone structures. Axioms A1--A3 and axioms A4--A5 play different roles in
our characterization result: the former
ones ensure sufficient richness of a given set family, while the latter ones
prevent it from being ``too rich.''

\begin{wrapfigure}{r}{4cm}
\vspace{-0.8cm}
\begin{center}
  \subfloat[{Before embedding.}]{%
  \resizebox{3cm}{!}{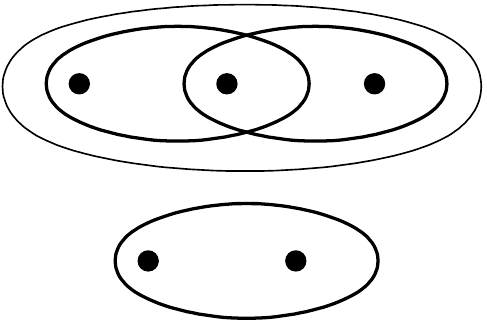}}\\
  \subfloat[{After embedding.}]{%
  \resizebox{3.5cm}{!}{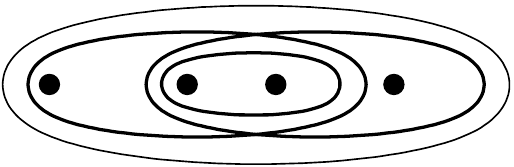}}
\end{center}
\vspace{-0.3cm}
\caption{\label{fig:embed}Clone structures from
  Example~\ref{exm:embed}.}
\vspace{-0.8cm}
\end{wrapfigure}

We will first build up the necessary tools for our inductive argument.
Let $\calE$ and $\calF$ be two families of subsets on two disjoint
finite sets $E$ and $F$, respectively.  We can embed $\calF$ into
$\calE$ as follows.  Given $e\in E$, let $\calE(e \rightarrow \calF)$
denote the family of subsets $\calE' \union \calF\subseteq
2^{(E\setminus\{e\})\cup F}$, where $\calE'$ is obtained from $\calE$
by replacing each set $X$ containing $e$ with $(X \setminus \{e\})\cup F$.

\begin{example}\label{exm:embed}
  Consider set families $\calD = \{\{x\},\{y\}, \{x,y\}\}$ and
  $\calC = \{\{a\},\{b\},\{c\},$ $\{a,b\},\{b,c\}, \{a,b,c\}\}$ 
  (both are strings of sausages and hence
  clone structures). Then, $\calC(b \rightarrow \calD)=
  \{ \{a\},\{x\},\{y\},\{x,y\}, \{c\}, \{a,x,y\}, \{x,y,c\},
  \{a,x,y,c\}\}$.
  It is easy to check that this, again, is a clone structure.
\end{example}

If $\calE$ and $\calF$ satisfy axioms A1--A5
then so does $\calE(e \rightarrow \calF)$. %
We prove it directly (it also follows from  Theorem~\ref{thm:char-main} 
combined with Proposition~\ref{pro:embedding}).

\begin{proposition}\label{prop:compose}
  Let $\calE$ and $\calF$ be families of subsets on disjoint sets $E$
  and $F$, respectively, that satisfy A1--A5. Then for any $e\in E$
  the set family $\calE(e \rightarrow \calF)$ also satisfies A1--A5.
\end{proposition}

Next, we define an inverse operation to embedding, 
which we call {\em  collapsing}. 
Observe that when we embed $\calF\subseteq 2^F$ into
$\calE\subseteq 2^E$, any $C\in \calE(e \rightarrow \calF)$ is either
a subset of $F$, a superset of $F$, or does not intersect $F$ at all.
Thus, for a set family $\calC$ on $A$ to be collapsible, it should
contain a set $A'$ that does not intersect non-trivially with any
other set in $\calC$.

\begin{definition}\label{def:subfamily}
  Let $\calF$ be a family of subsets on a finite set $F$.  A subset
  $\calE\subseteq\calF$ is a {\em subfamily} of $\calF$ if there
  is a set $E\in\calF$ such that
(i) $\calE=\{F\in\calF\mid F\subseteq E\}$;
(ii) for any $X\in \calF\setminus\calE$ %
    we have either $E\subseteq X$ or $X\cap E=\emptyset$.
  The set $E$ is called the {\em support} of $\calE$.  $\calE$ is
  called a {\em proper subfamily} of $\calF$ if $E$ is a proper subset~of~$F$.
\end{definition}
One can check that if $\calF$ satisfies A1--A5 and $\calE$ is a
subfamily of $\calF$, then $\calE$ satisfies A1--A5 as well. Note
that we require $E\in\calF$ (rather than just $E\subseteq F$), and
hence $E\in\calE$.

Let $\calF$ be a family of subsets on $F$ that satisfies A1--A5
and let $\calE$ be a proper subfamily of $\calF$ on $E\subset F$.
Then no set $Y\in\calF$ intersects $E$ non-trivially, and hence
$\calE$ can be ``collapsed''. That is, we can obtain a new set family
$\calB$ from $\calF$ by picking some alternative $b\notin F$, removing
all sets $X\in\calE\setminus\{E\}$ from $\calF$, and replacing each
set $Y$ that contains $E$ with $(Y\setminus E)\cup\{b\}$.  It is not
hard to check that ${\calB}$ satisfies A1--A5; the proof is similar
to that of Proposition~\ref{prop:compose}.  We will write
$\calF(\calE\rightarrow b)$ to denote the set family obtained by
collapsing a subfamily $\calE$ of $\calF$. That is, we have
$\calB=\calF(\calE\rightarrow b)$ if and only if $\calF = {\calB}(b
\rightarrow \calE)$.

Suppose that $\calF$ has no proper subfamilies; we will call such
subset families {\em irreducible}.  

\begin{theorem}
\label{thm:irreducible}
Any irreducible family of subsets satisfying A1--A5 is either a string
of sausages or a fat sausage.
\end{theorem}

Thus, any irreducible set family that satisfies A1--A5 is a clone
structure. This provides the basis for our inductive argument. For the
inductive step, we need to show that if $\calC$ and $\calD$ are two
clone structures over disjoint sets $C$ and $D$, and $c$ is some
candidate in $C$, then $\calC(c \rightarrow \calD)$ is a clone
structure. However, the proof of this fact is somewhat more
complicated than one might expect.  Indeed, suppose that we have a
pair of profiles $\calR=(R_1, \dots, R_n)$ and $\calQ=(Q_1, \dots,
Q_n)$ over sets $C$ and $D$, respectively, such that $\calC =
\calC(\calR)$ and $\calD = \calC(\calQ)$.  One might think that, given
$c\in C$, we can obtain a preference profile $\calR'$ such that
$\calC(c \rightarrow \calD)=\calC(\calR')$ simply by substituting
$Q_i$ for $c$ in $R_i$, for $i=1, \dots, n$. This intuition is not
entirely correct: without additional precautions, we may introduce
``parasite'' clones, i.e., clones that cross the boundary between $C$
and $D$.  However, we can construct $\calR'$, containing $n$
preference orders, from $\calR$ and $\calQ$ by tweaking this
construction slightly.

\begin{proposition}
\label{pro:embedding}
Let $\calC$ and $\calD$ be two clone structures over sets $C$ and $D$,
respectively, where $|C|=m$, $|D|=k$, and $C \cap D= \emptyset$.  Then
for each $c\in C$, the family of subsets $\calC(c \rightarrow \calD)$
is a clone structure.
\end{proposition}

\begin{theorem}\label{thm:char-main}
  A family $\calF$ of subsets of $[m]$ is a clone structure if and only
  if it satisfies conditions A1--A5.
\end{theorem}

Based on Theorem~\ref{thm:char-main}, it is easy to test
in polynomial time (and, in fact, even in logarithmic space) 
if a given set family 
(represented explicitly as a list of subsets) is a clone
structure: one simply needs to check if all the axioms hold.
We state this result more formally in Appendix~\ref{app:poly}.

\section{Compact Representations of Clone Structures}
\label{sec:representation}

Let us now consider the issue of representing clone structures.
We say that a clone structure 
$\calC$ is {\em $k$-implementable} if there is a $k$-voter profile
$\calR$ such that $\calC=\calC(\calR)$.
One might expect that to obtain a complex clone structure we
need an election with many voters. Yet,
each clone structure can be implemented by a profile with at
most three voters. 

\begin{theorem}\label{thm:3rep}
  Any clone structure is $3$-implementable.
\end{theorem}

Nonetheless, we would like a more structured representation. 
 In the
previous section we have seen that clone structures are organized
hierarchically and, thus, it is natural to represent them using
trees. The specific type of trees that are most convenient for this
task are PQ-trees introduced by Booth and
Lueker~\cite{boo-lue:j:consecutive-ones-property}.
A PQ-tree $T$ over a set $A = \{a_1, \ldots, a_n\}$ is an ordered tree 
that represents a family of permutations over $A$ as follows.  
The leaves of the tree correspond to the elements of $A$.
Each internal node is either of type P or of type Q. 
A {\em frontier} of $T$ is a permutation of $A$ obtained by reading 
the leaves of $T$ from left to right (recall that $T$ is ordered). 
The following operations are
allowed on the tree: If a node is of type P, then its children can be
permuted arbitrarily. If a node is of type Q, then the order of its
children can be reversed. A given permutation $\pi$ of $A$ is {\em  consistent} 
with a PQ-tree $T$, if we can obtain $\pi$ as the frontier of $T$
by applying the above operations.

We now describe a natural way to represent clone structures as
PQ-trees. Consider a clone structure $\calC$ over a finite set $C$.
Our characterization of irreducible clone structures implies that any
two proper irreducible subfamilies of $\calC$ have non-intersecting
supports.

\begin{proposition}\label{prop:unique}
  Let $\calC$ be a clone structure over a finite set $C$, and let
  $\calB$ and $\calD$ be two proper irreducible subfamilies of $\calC$
  on sets $B\subseteq C$ and $D\subseteq C$, respectively.  Then
  $B\cap D=\emptyset$.
\end{proposition}

Proposition~\ref{prop:unique} implies that every element of $C$
belongs to at most one proper irreducible subfamily of $\calC$.  Thus,
given a clone structure $\calC\subseteq 2^C$, there is a unique maximal
collection of pairwise disjoint sets $\Dec(C) = \{C_1, \dots, C_k\}$
such that $C_i\subseteq C$, $|C_i|\ge 2$, and for each $i=1, \dots, k$
the set family $\calC_i = \{C\in\calC\mid C\subseteq C_i\}$ is an
irreducible subfamily of $\calC$ (if $\calC$ is itself irreducible,
then $k=1$ and $C_1=C$). This collection can be efficiently constructed
by identifying the minimal (with respect to inclusion) non-singleton sets in
$\calC$: any such set of size $s\ge 3$ is itself an irreducible
clone structure (a fat sausage), and for a set of size $s=2$ we need to
find the maximal string of sausages that contains it. Note that it
need not be the case that $\bigcup_{i=1}^k C_i=C$; some elements may not
belong to any proper irreducible clone structure (consider, for
instance, the clone structure over $\{a, b, c, d\}$ given by $\{\{a\},
\{b\},\{c\},\{d\},\{b, c\},\{a, b, c, d\}\}$).  We will refer to the
collection $\Dec(C)$ as the {\em decomposition} of $\calC$.

\begin{wrapfigure}{r}{3.2cm}
\vspace{-0.5cm}
\centering
  \resizebox{3cm}{!}{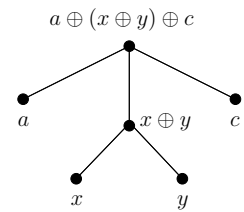}
  \caption{\label{fig:pqtree} Tree representation of the embedded
  clone structure from Example~\ref{exm:embed}.}
  \vspace{-0.3cm}
\end{wrapfigure}

We can now inductively define a PQ-tree $T(\calC)$ associated with a
clone structure $\calC\subseteq 2^C$ (for convenience, our PQ-trees
will be labeled).
Suppose first that $\calC$ is an
irreducible clone structure over the set $C=\{c_1, \dots, c_m\}$.
Then, by Theorem~\ref{thm:irreducible}, it is either a string of
sausages or a fat sausage. In the former case, assume without loss of
generality that $\calC$ is associated with the order $c_1\succ
c_2\succ \ldots\succ c_m$, i.e., it contains sets $\{c_i, c_{i+1}\}$
for $i=1, \dots, m-1$.  In both cases, we let $T(\calC)$ to be a tree
of depth $1$ that has $m$ (ordered) leaves. The $i$-th leaf is labeled
by $c_i$.  If $\calC$ is a string of sausages, the root of the tree is
of type Q and is labeled by $c_1\oplus\ldots\oplus c_m$; if $\calC$ is
a fat sausage, the root is of type P and is labeled by
$c_1\odot\ldots\odot c_m$. For $m=2$ the clone structure $\calC$
is both a string of sausages and a fat sausage; we treat it
as a fat sausage.

Now, if $\calC$ is reducible, we compute its decomposition
$\Dec(C)=\{C_1, \dots, C_k\}$.  For $i=1, \dots, k$, we set $\calC_i =
\{X\in\calC\mid X\subseteq C_i\}$, pick $c^1, \ldots, c^k\not\in C$,
and let $\calC'$ be the set family on the set $C'=(C\setminus
\bigcup_{i=1}^k C_i)\cup\{c^1, \dots, c^k\}$ given by
$
\calC'=\calC(\calC_1\rightarrow c^1, \ldots, \calC_k\rightarrow c^k). 
$
We then construct the tree $T(\calC')$.  This tree has leaves labeled
by $c^1, \dots, c^k$.  We replace each such leaf $c^i$ by the labeled
tree $T(\calC_i)$ for the irreducible set family $\calC_i$.

Given
the tree $T(\calC)$, we can reconstruct the clone structure $\calC$ in
an obvious way.
To illustrate this discussion, in Figure~\ref{fig:pqtree} we give a
PQ-tree for the clone structure from Example~\ref{exm:embed}.
We remark that the descendants of any internal node of $T(\calC)$ 
form a clone set. However, the converse is not necessarily true, i.e., 
there are clone sets that cannot be obtained in this way: 
if an internal node $v$ is labeled with a string of sausages and has $k$
children, $k\ge 3$, the descendants of any $\ell$ consecutive children
of $v$, $\ell<k$, form a clone set. Indeed, it is not hard to see
that any clone set corresponds either to a subtree of $T(\calC)$
or to a collection of subtrees of $T(\calC)$ 
whose roots are consecutive children of the same Q-node.

\section{Clones in Single-Peaked Elections}\label{sec:sp}
It is not unusual for voters to make 
their decisions based on the candidates' position 
on a single prominent issue, such as, e.g., the
level of taxation. Elections where all voters maike their decisions
in this way (with respect to the same issue) are known as {\em single-peaked}.
In such elections, the candidates are ordered with respect to their position on
the issue. This ordering is called \emph{the societal axis}; it can,
for example, order the candidates from those supporting the lowest level of taxation
to those supporting the highest level of taxation. Each voter $v$ forms her
preference order as follows. First, $v$ picks the candidate 
who is the closest to her ideal point on the societal axis. She then ranks
the remaining candidates according to her perceived distance
to the ideal point. The perception of distance may differ from one voter
to another: some voters may view a large deviation to the right
as less significant than a small deviation to the left, while others
may hold the opposite view. Thus, the voter will zig-zag
through the candidate list, ending her ranking with either the leftmost 
or the rightmost candidate.
For example, if the possible tax rates are $10\%$, $15\%$, $20\%$, and
$25\%$, and an election is single-peaked with respect to the axis
$10\% > 15\% > 20\% > 25\%$, a voter's  preference order 
may be, e.g., $20\% \succ 15\% \succ 10\% \succ 25\%$, but not 
$20\% > 10\% > 15\% > 25\%$.
Formally, a single-peaked preference profile is defined as follows.

\begin{definition}
  Let $\calR = (R_1, \ldots, R_n)$ be a profile over a candidate set
  $C$, and let $>$ be a linear order over $C$ (the societal axis).  We
  say that an order $\succ$ over $C$ is {\em compatible} with $>$ if
  for all $c,d,e \in C$ such that either $c > d > e$ or $e > d > c$, it holds that
  $c \succ d \implies d \succ e$.
  We say that $\calR$ is {\em single-peaked with respect to $>$}
  if each preference order in $\calR$ is compatible with $>$.
  A profile $\calR\in\calL^n(C)$ is called {\em single-peaked} if
  there exists a linear order over $C$ such that $\calR$ is
  single-peaked with respect to $>$; we say that $>$ {\em witnesses}
  single-peakedness of $\calR$.
\end{definition}

The literature on single-peaked elections is vast; 
for examples and intuition 
we point the
reader to the original paper of Black~\cite{bla:b:polsci:committees-elections}, 
which introduced this notion, and, for a more
algorithmic perspective, to some recent computational social choice 
papers~\cite{wal:c:uncertainty-in-preference-elicitation-aggregation,esc-lan-ozt:c:single-peaked-consistency,con:j:eliciting-singlepeaked,bra-bri-hem-hem:c:sp2,fal-hem-hem-rot:j:single-peaked-preferences}.

There are several reasons why single-peaked elections received so much
attention; the notion of single-peakedness is very natural, 
and indeed quite a few real-life elections are (close to) single-peaked. Further,
single-peaked elections have many desirable properties, of which
perhaps the best-known one is that they admit non-manipulable, non-dicatorial voting
rules (e.g., the median voter
rule~\cite{bla:b:polsci:committees-elections}).  This is a very
attractive property, which distinguishes single-peaked elections 
from those with unrestricted preferences 
(see~\cite{gib:j:polsci:manipulation,sat:j:polsci:manipulation}).

Unfortunately, if a candidate in a single-peaked election is cloned,
the election may lose the single-peakedness property: if many voters
find the clones very similar, they are likely to rank them randomly,
which may be incompatible with any societal axis. Thus, given an
election, we might want to check if it can be made single-peaked by
``decloning.'' Of course, we would like the resulting election to be
close to the original one.  Thus, we look for a single-peaked election
that ``collapses'' as few clones as possible, i.e., has the maximum
number of alternatives among all single-peaked elections that can be
obtained from the original one by decloning. The main result of this
section is a polynomial time algorithm for this problem.  Our
secondary goal is to understand which clone structures can arise in
single-peaked elections. We give a partial answer by, on one hand,
providing an example of a clone structure that cannot be implemented
by a single-peaked profile, and, on the other hand, identifying a
large family of clone structures that admit such an implementation.
\smallskip

\noindent\textbf{Preliminary Observations.}\quad
Let $\calR = (R_1, \ldots, R_n)$ be a preference profile over a
candidate set $C$, and let $D$ be some clone set in
$\calC(\calR)$. Given a $c\not\in C$, we write $\calR(D \mapsto c)$ to
denote the profile $\calR' = (R'_1, \ldots, R'_n)$, where each $R'_i
\in \calR'$ is obtained from $R_i$ by replacing the block of
candidates from $D$
with $c$.  We refer to the process of converting $\calR$ to $\calR(D
\mapsto c)$ as \emph{decloning $D$ to $c$ in $\calR$}. 
Note that we can declone $D$ even if the collection of subsets $\calD =
\{X\in\calC(\calR)\mid X\subseteq D\}$ is not a subfamily of
$\calC(\calR)$ (this is the only difference between decloning operation
$\mapsto$ and the collapsing operation $\rightarrow$).
Given a 
preference profile $\calR$ over a set of candidates $C$,
we let $c(\calR)=|C|$.

Given a profile $\calR = (R_1, \ldots, R_n)$, we refer to the
top-ranked candidate in $R_i\in\calR$ as the \emph{peak of $R_i$}, and
denote it by $\peak(R_i)$. We write $\peak(\calR)$ to denote the set
$\{ \peak(R_i) \mid R_i \in \calR\}$.  (Note that if $\calR$ is
single-peaked, this does not imply that $|\peak(\calR)|=1$:
intuitively, the term ``single-peaked'' refers to the shape of
individual preference orders with respect to a given societal axis.)
Theorem~\ref{thm:unique} and surrounding discussion in the appendix
clarifies the relation between possible societal axes and the set of
peaks of a profile.

We are ready to start our investigation of clones in single-peaked
profiles.  Consider a preference profile $\calR$ over $C$ that is
single-peaked  with respect to some order $>$, and let $D$ be a clone set with
respect to $\calR$. Do members of $D$ appear consecutively in $>$? Not necessarily
(take two votes, $b\succ c \succ a \succ d$ and $c\succ b \succ a \succ d$,
and axis $a > b > c > d$; $\{a,d\}$ is a clone even though $a$ and $d$ are
not consecutive in $>$), but they form at most two blocks within $>$.

\begin{proposition}
\label{thm:second-variety}
Let $\calR = (R_1, \ldots, R_n)$ be a preference profile over a
candidate set $C$ that is single-peaked  with respect to some order $>$, 
and let $D\in\calC(\calR)$ be a clone set for $\calR$ with $|D|\ge 2$.
Then $C$ can be partitioned
into pairwise disjoint sets $A_1$, $A_2$, $D_1$, $D_2$, and $P$ so that 
$C\setminus D=A_1\cup P\cup A_2$, $D=D_1\cup D_2$, 
$D_1\neq\emptyset$, $D_2\neq\emptyset$, and
$
A_1 > D_1 > P > D_2 > A_2.
$
Further, if $P\neq\emptyset$, then $\peak(\calR)\subseteq P$, and, moreover,  
$P \succ_i D \succ_i A_1 \union A_2$ for each $i=1, \dots, n$.
\end{proposition}

Proposition~\ref{thm:second-variety} motivates a very useful
classification of clone sets in single-peaked profiles.
  Let $\calR$ be a profile over some candidate set $C$ that is single-peaked
   with respect to some order $>$. If the elements of $D$  are ranked contiguously in $>$, 
  then we say that $D$ is a {\em clone set of the first type} with respect to $>$; 
  otherwise, we say that $D$ is a {\em clone set of the second type} with respect to $>$.
The following is an immediate corollary of the second part of 
Proposition~\ref{thm:second-variety}.

\begin{corollary}\label{cor:2type-nopeaks} 
  Let $\calR$ be a single-peaked preference
  profile. %
  If $D \in \calC(\calR)$ is a clone set of the second type
  w.r.t.~some societal axis $>$, then $D$ does not contain any peaks
  of $\calR$.
\end{corollary}

Observe that if $D$ is a clone set of the first type that does not contain
any peaks of $\calR$, then for each voter $i$ either $\peak(R_i)\succ_i D$, 
in which case $\succ_i$ coincides with $>$ on $D$, or $D\succ_i \peak(R_i)$, 
in which case $\succ_i$ coincides with $\ola{>}$ on $D$. Thus,
the following corollary.

\begin{corollary}\label{cor:2types}
  Let $\calR$ be a single-peaked preference profile over a candidate set $C$.
  If $D \in \calC(\calR)$ is a clone set of the first type with respect to some
  societal axis $>$, $|D|\ge 2$, and $D\cap\peak(\calR)=\emptyset$, 
  then $D$ is a string of sausages.
\end{corollary}

\noindent\textbf{Decloning Towards a Single-Peaked Profile.}\quad
We will now present an algorithm for transforming a given election
into a single-peaked one by decloning.  Our algorithm works with a
PQ-tree $T$ that captures the clone structure of our profile.
Informally, it first contracts the tree to a single node, and then
greedily reintroduces clone sets, following the branches of the tree,
while maintaining the invariant that the resulting profile is
single-peaked (this requires care as single-peakedness of a profile
can be witnessed by many different societal axes).  In what follows,
we present a version of our algorithm that only declones clone sets
that correspond to subtrees of $T$. This algorithm produces the
optimal decloning for many settings (and, in particular, for all
profiles whose PQ-trees contain P-nodes only), but may fail to find an
optimal solution in some cases. We will show an example of such a case
and outline a more sophisticated polynomial-time algorithm that works
for all profiles.
We describe our algorithm in terms of proper colorings of a
PQ-tree.

\begin{definition}
  Let $T$ be a PQ-tree. A {\em coloring} of $T$ is
  a function $f$ from the set of nodes of $T$ to
  the set $\{\mathit{black}, \mathit{white}\}$.
  A coloring of $T$ is {\em proper} if the children of each black node are black.
  Given a coloring $f$ of $T$, let $W(f)$ denote the set of nodes
  that are $\mathit{white}$ under $f$.
\end{definition}

\begin{definition}\label{def:d2c}
  Let $\calR$ be a profile over a candidate set $C$, set $T =
  T(\calC(\calR))$, and let $f$ be a proper coloring of $T$.  For a
  node $v\in T$, let $C_v = \{c \in C \mid c$ is a leaf of $T$'s
    subtree rooted in  $v \}$.  Define $\calR(T,f)$ to be the profile
  obtained from $\calR$ as follows: for each internal node $v$, if $v$
  is black and its parent is white (or $v$ is the root), declone the
  set $C_v$ to a single candidate $v$.
\end{definition}

Our algorithm {\sc BasicDecloneSP} takes as an input  
a preference profile $\calR$. It then constructs
a PQ-tree $T = T(\calC(\calR))$ 
that corresponds to the clone structure of $\calR$.
It initializes $f$ to be a coloring of $T$ in which every node is $\mathit{black}$;
this coloring will be modified during the execution of the algorithm.  
Note that at this point $\calR(T,f)$ is single-peaked and $c(\calR(T, f))=1$. 
The algorithm also maintains a queue of nodes that it intends to visit.
Initially, the queue contains the root of $T$ only.
Throughout the execution, we ensure that $\calR(T,f)$ is single-peaked
and all ancestors of each node in the queue are white;
note that this is indeed the case just after the initialization stage.

At each stage, {\sc BasicDecloneSP}$(\calR)$ picks a node $v$ from the
queue (if the queue is empty, the algorithm terminates).  It then
executes the following steps:
  \begin{enumerate}%
  \item Set $f(v) = \mathit{white}$.
  \item Check if $\calR(T,f)$ is single-peaked (possible in polynomial
    time~\cite{bar-tri:j:stable-matching-from-psychological-model,esc-lan-ozt:c:single-peaked-consistency}).
  \item If $\calR(T,f)$ is single-peaked, add all children of $v$ to the queue.
        Otherwise, reset $f(v) = \mathit{black}$.
  \end{enumerate}

By induction on the execution of the algorithm, one can see
that at each point in time $f$ is a proper coloring of $T$, 
and therefore $\calR(T,f)$ is well-defined. Further, each node is processed
at most once, so {\sc BasicDecloneSP}$(\calR)$ runs in polynomial time.
Finally, it is clear that {\sc BasicDecloneSP}$(\calR)$ produces
a single-peaked election.

We now show that
given a profile $\calR$ over a candidate set $C$, {\sc BasicDecloneSP}
outputs a single-peaked profile $\calR'$ with the following property:
$c(\calR')\le c(\calR'')$ for any single-peaked preference profile
$\calR''$ that can be obtained from $\calR$ by decloning clone sets
that correspond to subtrees of $T(\calC(\calR))$.  
First, as a sanity check, we note (Proposition~\ref{thm:declone-to-sp}
in the appendix) that if a profile is already single-peaked then it
remains single-peaked after decloning. Second, we show that the greedy
way in which {\sc BasicDecloneSP} reintroduces clones cannot prevent
us from finding an optimal solution.

\begin{proposition}\label{cor:greedy} 
  Let $\calR$ be a preference profile over a set of candidates $C$.
  Let $D^1, \ldots, D^k \in \calC(\calR)$ be a sequence of
  pairwise disjoint clone sets, and let $c^1, \ldots, c^k$ be
  a sequence of distinct candidates not in $C$. For each $i=1, \dots, k$,
  let $\calR_i$ denote a preference profile in which for
  each $j=1, \dots, k$, $j \neq i$, $D^j$ is decloned to $c^j$.
  Then  $\calR$ is single-peaked if and only if each of the profiles
  $\calR_i$, $i=1, \dots, k$, is single-peaked.
\end{proposition}
The idea of the proof is to show that if reintroducing a given clone
does not break single-peakedness of a profile, then one can obtain a
societal axis witnessing this fact by local modifications of the
societal axis used prior to the clone's reintroduction. In the next
theorem we show correctness of {\sc BasicDecloneSP} for the case
of decloning subtrees only, and in the following proposition we show
that sometimes decloning subtrees does not suffice.

\begin{theorem}\label{thm:declone-alg}
  Given a preference profile $\calR$ over a candidate set $C$, 
  {\sc BasicDecloneSP}$(\calR)$ runs in time polynomial in $|\calR|$ and $|C|$,
  and produces a single-peaked preference profile $\calR'$ 
  such that $c(\calR')\le c(\calR'')$ for any single-peaked preference profile $\calR''$
  that can be obtained from $\calR$ by decloning one or more clone sets
  that correspond to subtrees of $T(\calC(\calR))$. 
\end{theorem}

\begin{proposition}
\label{thm:counterexample}
Let $\calC$ be a string of sausages over candidates $\{a,b,c\}$, 
let $\calD'$ be a string of sausages over candidates $\{1,2,3\}$, 
and let $\calD''$ be a fat sausage over candidates $\{x, y\}$.
Then the clone structure $\calC'=\calC(b \rightarrow \calD')$ is not
single-peaked, but the clone structure 
$\calC''=\calC(b \rightarrow \calD'')$ is single-peaked.
\end{proposition}

In Proposition~\ref{thm:counterexample} the clone structure $\calC''$
can be obtained from $\calC'$ by decloning the clone $\{1, 2\}$ to $x$
(and renaming $3$ as $y$). Further, $\calC''$ has four candidates, whereas
any clone structure that can be obtained from $\calC'$ by decloning
clone sets that correspond to substructures of $\calC$ will have at most 
three candidates. This shows that if $T(\calC(\calR))$ contains Q-nodes, 
{\sc BasicDecloneSP}$(\calR)$ may fail to find
the optimal decloning of a given profile into a single-peaked one
(note, however, that this issue does not arise if $T(\calC(\calR))$
contains P-nodes only, in which case {\sc BasicDecloneSP}$(\calR)$
actually finds an optimal solution).

Hence, to obtain an algorithm that {\em always} finds an optimal decloning, 
we need to modify {\sc BasicDecloneSP}$(\calR)$ to also consider 
clone sets that correspond to substrings of strings of sausages. 
However, a straightforward implementation of this idea leads
to an exponential blow-up in the running time: there are exponentially many 
ways to choose a non-overlapping collection of substrings of a given string.
Fortunately, it turns out that it suffices to consider breaking 
strings of sausages into two parts. In Section~\ref{app:declone-alg} of the appendix 
we present a polynomial-time algorithm that makes use of this idea
and always finds an optimal single-peaked decloning of a given preference profile;
this algorithm differs from {\sc BasicDecloneSP}$(\calR)$
in its handling of Q-nodes only.
\smallskip

\noindent\textbf{Clone Structures via Single-Peaked Profiles.}\quad
Let us now turn to our second goal, the problem of characterizing
clone structures that can be implemented using single-peaked profiles.
Formally, we say that a clone structure $\calC$ is {\em single-peaked}
if there exists a single-peaked profile $\calR$ such that $\calC =
\calC(\calR)$. 

\begin{proposition}\label{prop:irr-sp}
  Fat sausages and strings of sausages are single-peaked.
\end{proposition}

Thus, by Proposition~\ref{thm:counterexample}, 
single-peaked clone structures are not closed under embeddings.
Nonetheless, we identify a large class of single-peaked clone
structures.

\begin{proposition}\label{prop:sp-tree}
  Let $\calC$ be a clone structure over a set of candidates $C$. Suppose that
  for each Q-node of the PQ-tree decomposition $T(\calC)$ of $\calC$
  it holds that all children, except possibly the leftmost child and the
  rightmost child, are labeled with singletons, i.e., elements of $C$.
  Then $\calC$ is a single-peaked clone structure.
\end{proposition}

The second part of Proposition~\ref{thm:counterexample} shows that
Proposition~\ref{prop:sp-tree} does not characterize single-peaked
clone structures; finding such a characterization is an interesting
open problem.

\section{Clones in Single-Crossing Elections}
\label{sec:sc}

Let us now consider a different domain restriction called
\emph{single-crossing}. The idea behind single-crossing elections is
similar to that behind single-peaked elections, but now it is the
voters who are ordered along some axis; say, the traditional
left-to-right spectrum of political views. Consider a voter $v$ 
on one of the extreme ends of the spectrum and
two candidates, $c$ and $d$ such that
$v$ prefers $c$ to $d$. As we move toward the other end of the voter
spectrum, for a while voters agree that $c$ is better than $d$, but
eventually $d$ \emph{crosses} $c$ and, from this point on, the voters
prefer $d$ to $c$. Single-crossing dates back at least to
Mirrlees~\cite{mir:j:single-crossing}; see
also~\cite{gan-sma:j:single-crossing,sap-toh:j:single-crossing-strategy-proofness,bar-mor:j:top-monotonicity}
for more recent work that also describes realistic settings where
single-crossing profiles arise. Formally, we use the following
definition.

\begin{definition}\label{def:sc}
  We say that a preference profile $\calR = (R_1, \ldots, R_n)$ over
  candidate set $C$ is {\em single-crossing with respect to order $\lhd$ over
  $[n]$}, if for every pair of distinct candidates $c, d \in C$ it holds
  that either 
  $\{i \mid c \succ_i d \} \lhd \{j \mid d \succ_j c \}$ or 
  $\{i \mid d \succ_i c \} \lhd \{j \mid c \succ_j d \}$.  We say that a profile is
  {\em single-crossing} if there exists an order $\lhd$ with respect to which
  it is single-crossing.
\end{definition}

Strictly speaking, the notion introduced in Definition~\ref{def:sc} is
referred to as \emph{order restriction} and not
single-crossing. However, in our setting these two notions are
equivalent and the term ``single-crossing'' much more intuitively
describes the notion.  

We observe that one can check in polynomial time whether a given profile is single-crossing;
to the best of our knowledge, this observation does not appear in the literature.
  
\begin{proposition}\label{prop:sc-easy}
The problem of checking if a given profile is single-crossing is in $\p$.
\end{proposition}

As in the case of single-peakedness, we would like to know which clone
structures can be implemented by single-crossing profiles and what
is the complexity of decloning towards a single-crossing
profile. We reach opposite answers from those for the
case of single-peakedness.

\begin{theorem}\label{thm:scimp}
For every clone structure $\calC$ there exists a single-crossing profile
$\calR$ such that $\calC = \calC(\calR)$.
\end{theorem}

\noindent 
We remark that, unlike the construction in the proof of Proposition~\ref{pro:embedding}, 
which leads to a $3$-voter profile (Theorem~\ref{thm:3rep}),  
the proof of Theorem~\ref{thm:scimp} produces a profile with many voters. 

\begin{theorem}\label{thm:scnp}
  Given a profile $\calR$ over a candidate set $C$ and a positive
  integer $k$, it is $\np$-hard to decide if there exists a
  single-crossing profile $\calR'$ with $c(\calR') \geq k$ 
  that can be obtained from $\calR$ by decloning.
\end{theorem}

\noindent However, optimal decloning is easy if the order of voters is fixed.

\begin{proposition}\label{prop:sc-fixedeasy}
  Given a profile $\calR$ over a candidate set $C$, a positive
  integer $k$, and an order $\lhd$, we can decide in polynomial time if there exists a
  profile $\calR'$ with $c(\calR') \geq k$ that is single-crossing with respect to $\lhd$
  and can be obtained from $\calR$ by decloning.
\end{proposition}

\section{Conclusions and Future Work}\label{sec:concl}
We have characterized clone structures in elections, obtained a
convenient representation using PQ-trees, and used this representation
in an algorithm that restores an election's single-peakedness by
decloning as few candidates as possible.  On the other hand, we have shown that
recovering the single-crossing property optimally is $\np$-hard. We also
made first steps toward characterizing clone structures in
single-peaked elections and we have shown that all clone structures
can be implemented with single-crossing profiles.

Other research directions include establishing the complexity of
verifying whether a given candidate can be made an election winner
(under a particular voting rule) by decloning a given number of candidates.
We hope that our work will facilitate obtaining hardness
results for this problem, thus complementing the easiness results of
Elkind et al.~\cite{elk-fal-sli:c:cloning}.

\bibliographystyle{abbrv}
\bibliography{grystacs}

\appendix

\renewcommand{\thetheorem}{\Alph{section}.\arabic{theorem}}

\section{Material Missing from  Section~\ref{sec:char}}
\label{app:char}

Before we move on to the proofs, we need several additional definitions.

\begin{definition}
Let $C$ be a set of candidates.
\begin{enumerate}
\item Given three pairwise disjoint subsets $X, Y, Z$ of $C$, and an
  order $\succ$ over $C$, we say that $X$ {\em separates} $Y$ and $Z$
  in $\succ$ if either $Y \succ X \succ Z$ or $Z \succ X \succ Y$.
\item We say that an alternative $a \in C$ {\em splits} a subset $X
  \subseteq C$ with respect to an order $\succ$ if $X$ can be
  partitioned into two nonempty sets $X_1$ and $X_2$ such that $\{a\}$
  separates $X_1$ and $X_2$ in $\succ$; note that this implies
  $a\not\in X$.
\end{enumerate}
\end{definition}

\begin{proposition}
\label{pro:reversing}
Given a profile $\calR=(R_1, \dots, R_n)$, let $\calR'=(R'_1, \dots,
R'_n)$ be a profile such that $R'_i\in\{R_i, \ola{R_i}\}$ for all
$i=1, \dots, n$.  Then ${\calC}(\calR)={\calC}(\calR')$.
\end{proposition}

\newtheorem*{propositionproinclusive}{Proposition~\ref{pro:inclusive}}

\begin{propositionproinclusive}
Let $\calR$ be a profile on $[m]$. Then
(1) $\{i\} \in {\calC}(\calR)$ for any $i\in [m]$;
(2) $\emptyset \notin \calC(\calR)$ and $[m] \in {\calC}(\calR)$;
(3) if $C_1$ and $C_2$ are in ${\calC}(\calR)$ and $C_1\cap
  C_2\ne \emptyset$, then $C_1\cup C_2$ and $C_1\cap C_2$ are also in
  ${\calC}(\calR)$;
(4) if $C_1$ and $C_2$ are in $\calC(\calR)$ and $C_1 \bowtie
  C_2$, then $C_1 \setminus C_2$ and $C_2 \setminus C_1$ are also in
  $\calC(\calR)$.
\end{propositionproinclusive}

\begin{proof}%
  Properties~(1) and~(2) are immediate.  Let us prove~(3). Let $C_1$
  and $C_2$ be two sets in ${\calC}(\calR)$ with $I=C_1\cap C_2\ne
  \emptyset$, and let $R_i$ be an arbitrary preference order from
  $\calR$.  Consider some $a \in [m]$. Since $C_1 \in \calC(\calR)$,
  if $a \in [m] \setminus C_1$, then $a$ does not split $C_1$ and, as
  a result, $a$ does not split $I$.  Similarly, no alternative in
  $[m]\setminus C_2$ can split $I$. Thus, no element outside of $I$
  can split $I$, and hence members of $I$ are ranked contiguously in
  $R_i$.  Since this holds for any $R_i$ in $\calR$, we have
  $I\in\calC(\calR)$.

  Now, suppose that there is an alternative $a \in [m] \setminus (C_1
  \cup C_2)$ that splits $C_1 \cup C_2$ in some order $R_i$.  
  We know that $a$ splits neither $C_1$ nor $C_2$, hence
  either $C_1 \succ_i a \succ_i C_2$ or $C_2 \succ_i a \succ_i C_1$,
  which is impossible since the intersection of $C_1$ and $C_2$ is
  nonempty.  Thus, $C_1 \union C_2 \in \calC(\calR)$. This proves
  that~(3) holds.

  Let us now consider property~(4). Suppose $C_1, C_2 \in
  \calC(\calR)$ and $C_1\bowtie C_2$.  Consider the set $C_1 \setminus
  C_2$; for $C_2\setminus C_1$ the argument is similar.  First, no
  element outside of $C_1$ can split $C_1 \setminus C_2$, because
  otherwise it will split $C_1$ too.  Further, in each $R_i$, the
  intersection $C_1 \cap C_2$ separates $C_1 \setminus C_2$ and $C_2
  \setminus C_1$.  Hence elements of $C_1 \cap C_2$ cannot split $C_1
  \setminus C_2$ either, and the property follows.
\end{proof}

\newtheorem*{proexclusive1}{Proposition~\ref{pro:exclusive1}}

\begin{proexclusive1}
For any profile $\calR$ on $[m]$, each 
$X \in \calC(\calR)$ 
has at most two proper minimal supersets in $\calC(\calR)$.
\end{proexclusive1}
\begin{proof}%
  For the sake of contradiction, assume that there are three distinct
  sets $Y, Z, W$ in $\calC(\calR)$ such that each of them is a proper
  minimal superset of $X$. It is easy to see that $Y \cap Z = X$: 
  by Proposition~\ref{pro:inclusive}, $Y\cap Z\in\calC(\calR)$, so if
  $(Y \cap Z)\setminus X\neq\emptyset$, neither $Y$ nor $Z$ would be a
  proper minimal superset of $X$.  Similarly, $Y \cap W = X$ and $Z
  \cap W = X$.  Pick two alternatives $y,z$ so that $y \in Y\setminus
  X$ and $z \in Z\setminus X$.  Let $R_i$ be a preference order from
  $\calR$.  The set $X$ separates $Y \setminus X$ and $Z \setminus X$,
  and so either $y \succ_i X \succ_i z$ or $z \succ_i X \succ_i y$; by
  Proposition~\ref{pro:reversing} we may assume the former.  Now, pick
  $w \in W\setminus X$. A similar argument shows that we have $y
  \succ_i X \succ_i w$ (as $w \succ_i X \succ_i y$ leads to a
  contradiction). But now we must have $z \succ_i X \succ_i w$ or $w
  \succ_i X \succ_i z$, none of which is possible.
\end{proof}

\newtheorem*{proexclusive2}{Proposition~\ref{pro:exclusive2}}
\begin{proexclusive2}
  If $\calC$ is a
  clone structure, %
  it does not contain a bicycle chain.
\end{proexclusive2}
\begin{proof}%
  Suppose that a clone structure $\calC$
  contains a bicycle chain $\{A_0, \dots, A_{k-1}\}$, and let
  $\calR=(R_1, \dots, R_n)$ be a preference profile such that
  $\calC=\calC(\calR)$.

  As argued in the proof of Proposition~\ref{pro:inclusive}, the set
  $A_0\cap A_1$ separates $A_0\setminus A_1$ and $A_1\setminus A_0$ in
  $R_1$.  Thus, by Proposition~\ref{pro:reversing}, we can assume that
  we have
  $
  A_0\setminus A_1\succ_1 A_0\cap A_1\succ_1 A_1\setminus A_0.
  $ 
  Further, by the definition of the bicycle chain we have
  $A_1\setminus A_0= A_1\cap A_2$, $A_1\setminus A_2= A_0\cap A_1$.
  Again, we have
  $
  A_1\cap A_0 = A_1\setminus A_2\succ_1 A_1\cap A_2\succ_1 A_2\setminus A_1.
  $ 
  Now, if $k=3$, we have a contradiction already: since $A_0\setminus
  A_1\neq\emptyset$ and $A_0\subseteq A_1\cup A_2$, it has to be the
  case that $A_0$ intersects $A_2$, yet all elements of $A_2$ are
  ranked strictly below $A_0$.  If $k>3$, continuing inductively, we
  obtain that for each $i=1, \dots, k-1$ the set $A_i$ is ranked below
  $A_{i-1}\setminus A_i$ in $R_1$. Hence, all elements of $A_{k-1}$
  are ranked below $A_0$ in $R_1$.  However, we have $A_0\cap
  A_{k-1}\neq\emptyset$, a contradiction.
\end{proof}

\newtheorem*{propcompose}{Proposition~\ref{prop:compose}}
\begin{propcompose}
  Let $\calE$ and $\calF$ be families of subsets on disjoint sets $E$
  and $F$, respectively, that satisfy A1--A5. Then for any $e\in E$
  the set family $\calE(e \rightarrow \calF)$ also satisfies A1--A5.
\end{propcompose}
\begin{proof}%
We have $\{e'\}\in \calE$ for all $e'\in E$, $\{f\}\in \calF$ for
  all $f\in F$, so $\{g\}\in \calE(e \rightarrow \calF)$ for all $g\in
  (E\setminus\{e\})\cup F$. %
  Clearly, $\emptyset\not\in \calE(e \rightarrow \calF)$.  Further,
  $E\in\calE$ and $e\in E$, so $(E\setminus\{e\})\cup F\in \calE(e
  \rightarrow \calF)$. Thus, A1 is satisfied.

  Throughout the rest of the proof, we will use the observation that
  no set $D\in \calE(e \rightarrow \calF)$ can intersect $F$
  non-trivially, i.e., we have that for each $D\in \calE(e
  \rightarrow \calF)$ it holds that either $D \cap F = \emptyset$
  or $D \setminus F = \emptyset$ or $F \setminus D = \emptyset$.

  We will now show that $\calE(e \rightarrow \calF)$ satisfies A3 and
  A4.  Consider two sets $C_1, C_2$ in $\calE(e \rightarrow \calF)$
  such that $C_1\cap C_2\neq\emptyset$. If $C_1\subseteq C_2$ or
  $C_2\subseteq C_1$, then A2 and A3 trivially hold, so we can assume
  that $C_1\bowtie C_2$.
 
  Suppose first that $C_1\subseteq F$. Then $C_2\cap F\neq\emptyset$
  and it cannot be the case that $F\subseteq C_2$, since we assume
  $C_1\setminus C_2\neq\emptyset$.  Hence, $C_2\subseteq F$, so $C_1,
  C_2\in\calF$, and the sets $C_1\cap C_2, C_1\cup C_2, C_1\setminus
  C_2, C_2\setminus C_1$ belong to $\calF$ and hence to $\calE(e
  \rightarrow \calF)$.

  Next, suppose that $F\subseteq C_1$. Set $C'_1=(C_1\setminus
  F)\cup\{e\}$.  Since $C_2\not\subseteq C_1$, we have
  $C_2\not\subseteq F$, and hence either $F\subseteq C_2$ or $F\cap
  C_2=\emptyset$.  In the former case, set $C'_2=(C_2\setminus
  F)\cup\{e\}$; in the latter case, set $C'_2=C_2$. In both cases, we
  have $C'_1, C'_2\in \calE$, and $C'_1\bowtie C'_2$.  Therefore, the
  sets $C'_1\cap C'_2, C'_1\cup C'_2, C'_1\setminus C'_2,
  C'_2\setminus C'_1$ belong to $\calE$, and hence the sets $C_1\cap
  C_2, C_1\cup C_2, C_1\setminus C_2, C_2\setminus C_1$ belong to
  $\calE(e \rightarrow \calF)$.

  If $F\subseteq C_2$, the argument is similar. Thus, it remains to
  consider the case $C_1 \cap F=\emptyset$, $C_2 \cap F=\emptyset$.
  Then $C_1, C_2\in\calE$. Thus, the sets $C_1\cap C_2, C_1\cup C_2,
  C_1\setminus C_2, C_2\setminus C_1$ belong to $\calE$ and do not
  contain $e$, and hence they belong to $\calE(e \rightarrow \calF)$
  as well.  Thus, axioms A2 and A3 are satisfied.

  To show that A4 holds, assume for the sake of contradiction that
  some set $C\in \calE(e \rightarrow \calF)$ has three proper minimal
  supersets $X$, $Y$, and $Z$ in $\calE(e \rightarrow \calF)$.  If we
  have $C\subseteq F$, then the sets $X$, $Y$ and $Z$ cannot strictly
  contain $F$ (or they would not be proper minimal supersets), but
  have to intersect $F$, so it has to be the case that $X, Y,
  Z\subseteq F$. Thus, $C$ has three proper minimal supersets in
  $\calF$, a contradiction.  Next, suppose that $F\subseteq C$. Then
  all three sets $X$, $Y$ and $Z$ are supersets of $F$, too. Consider
  the sets $C'=(C\setminus F)\cup\{e\}$, $X'=(X\setminus F)\cup\{e\}$,
  $Y'=(Y\setminus F)\cup\{e\}$, $Z'=(Z\setminus F)\cup\{e\}$.  All
  these sets are in $\calE$. Moreover, $X'$, $Y'$ and $Z'$ are
  distinct, and each of them is a superset of $C'$. To see that all of
  them are proper supersets of $C'$, observe that if $C'\subset
  T'\subset X'$, then $C\subset (T'\setminus\{e\})\cup F\subset X$, a
  contradiction with $X$ being a minimal proper superset of $C$.
  Thus, $C'$ has three minimal proper supersets in $\calE$, a
  contradiction.  Finally, if $C\cap F=\emptyset$, we have
  $C\in\calE$.  For $T=X, Y, Z$, let $T'=T$ if $F\cap T=\emptyset$ and
  $T'=(T\setminus F)\cup \{e\}$ otherwise. Clearly, the sets $X'$,
  $Y'$ and $Z'$ are in $\calE$.  By the same argument as above, we can
  show that $C$ has three minimal proper supersets in $\calE$, a
  contradiction. Thus, $\calE(e \rightarrow \calF)$ satisfies A4.

  Finally, let $\calE(e \rightarrow \calF)$ contain a
  bicycle chain $\{A_0, \dots, A_{k-1}\}$; in what follows, all
  indices are computed modulo $k$.  Suppose first that we have
  $A_i\subseteq F$ for some $i=0, \dots, k-1$. Then $A_{i-1}\cap
  F\neq\emptyset$, $A_{i+1}\cap F\neq\emptyset$. Since both $A_{i-1}$
  and $A_{i+1}$ intersect $A_i$ non-trivially, neither of them can
  contain $F$, and therefore both of them are subsets of $F$.
  Applying this argument inductively, we conclude that all sets $A_i$,
  $i=0, \dots, k-1$, are subsets of $F$, i.e., $\calF$ contains a
  bicycle chain, a contradiction.
  Thus, we can assume that for each $i=0, \dots, k-1$ either
  $F\subseteq A_i$ or $F\cap A_i=\emptyset$.  For each $i=0, \dots,
  k-1$, set $A'_i=(A_i\setminus F)\cup\{e\}$ if $F\subseteq A_i$ and
  set $A'_i=A_i$ otherwise. It is straightforward to check that the
  set family $\{A'_0, \dots, A'_{k-1}\}$ is a bicycle chain in
  $\calE$, a contradiction. Thus, $\calE(e \rightarrow \calF)$
  satisfies A5.
\end{proof}

\begin{proposition}\label{pro:two}
Let $\calF$ be an irreducible family of subsets of $[m]$ that
satisfies A1--A5, and let $D$ be a minimal proper subset of
$\calF$. Then $|D|=2$.
\end{proposition}
\begin{proof}%
  Suppose for the sake of contradiction that $|D| \geq 3$. The set
  family ${\calD}=\{F\in\calF\mid F\subseteq D\}$ is not a subfamily
  of $\calF$, which means that $\calF$ contains a proper subset $E$
  such that $D\bowtie E$.  However, by A2 and A3, both $D \cap E$ and
  $D \setminus E$ must belong to $\calF$, both are strict subsets of
  $D$, and at least one of them has at least two elements.  Thus, $D$
  is not a minimal proper subset, a contradiction.
\end{proof}

\begin{proposition}\label{pro:pairs}
Let $\calF$ be a family of subsets of $[m]$ that satisfies A1--A5.
Then each candidate $i \in [m]$ belongs to at most two minimal
proper subsets in $\calF$.
\end{proposition}
\begin{proof}%
  Suppose for the sake of contradiction, that $i$ belongs to three
  minimal proper subsets in $\calF$.  Since these subsets are minimal
  proper subsets, they are also minimal proper supersets of
  $\{i\}$. However, by A4, no subset of $\calF$ has more than two
  minimal proper supersets, a contradiction.
\end{proof}

\newtheorem*{thmirreducible}{Theorem~\ref{thm:irreducible}}

\begin{thmirreducible}
Any irreducible family of subsets satisfying A1--A5 is either a string
of sausages or a fat sausage.
\end{thmirreducible}
\begin{proof}%
  Let $\calF$ be an irreducible family of subsets over $[m]$ that
  satisfies A1--A5.  If $\calF$ does not contain any proper subsets,
  then it is a fat sausage.  Thus, for the remainder of the proof let
  us assume that $\calF$ does contain at least one proper subset.

  Let us consider a graph $G$ whose vertices are elements of $[m]$ and
  there is an edge between $i$ and $j$ if and only if $\{i,j\}$ is a
  minimal proper subset of ${\calF}$.  By
  Proposition~\ref{pro:pairs}, the degree of each vertex in $G$ is at
  most $2$. Further, $G$ cannot contain cycles, since each cycle in
  $G$ would correspond to a bicycle chain in $\calF$ formed by the
  two-element subsets $\{i,j\}$. Thus, $G$ is a collection of paths.
  We will now prove that $G$ has at most one connected component, and
  hence $\calF$ is a string of sausages.

  Let $G'$ be a maximal connected component in $G$, and let $F$ be the
  set of vertices of $G'$.  Suppose that $F\neq [m]$.  Note that by
  A3, $F$ is a subset in $\calF$.  Since $\calF$ is not a fat sausage,
  by Proposition~\ref{pro:two} we have $|F|\ge 2$.
  Let us rename the alternatives so that $F = \{f_1, \ldots, f_k\}$
  and each $\{f_i,f_{i+1}\}$, $1 \leq i < k$, is an edge of $G'$.

  If $F\ne [m]$, there exists a proper subset $E\in \calF$ such that
  $E\bowtie F$.  Let us pick such a set $E$ for which $|E \setminus
  F|$ is smallest.  By A3, the set $E \setminus F$ belongs to $\calF$.
  We consider two cases.

\medskip
\noindent{\bf $\mathbf{|\mathbf{E} \setminus F| = 1}.$\quad} 
    Observe that in this
    case $|E\cap F|\ge 2$: otherwise, $E$ would be an edge of $G'$.
    Let $e$ be a member of $E\setminus F$.  Suppose first that $E\cap
    F$ is not a contiguous subset of $F$, that is, there are some
    $i,j,\ell \in [k]$ such that $i < \ell < j$, and (i) $f_i \in E$
    and $f_s\not\in E$ for $s<i$, (ii) $f_j \in E$ and $f_t\not\in E$
    for $t>j$, and (iii) $f_\ell \notin E$.  Then either $E\cap
    F=\{f_i, f_j\}$, or $E \cap F$ intersects $\{f_{i+1}, \ldots,
    f_{j-1}\}$ and we have $\{f_i,f_j\} = (E \cap F)\setminus
    \{f_{i+1}, \ldots, f_{j-1}\}$.  In both cases, we can use axiom A3
    to conclude that $\{f_i,f_j\}$ belongs to $\calF$, and hence $G'$
    contains a cycle, a contradiction.  Thus, we have $E \intsct
    F=\{f_i, \dots, f_j\}$ for some $1\le i<j\le k$.
 
    Suppose that $j\neq k$. Then, since $i<j$, by A3 the set $E
    \setminus \{f_1, \ldots, f_{j-1}\} = \{e,f_j\}$ is in $\calF$.
    However, this means that $e \in F$, which is a contradiction.
    Thus, $j=k$. Similarly, we can argue that $i=1$.  Hence, we have
    $F\subseteq E$, a contradiction.

\medskip
\noindent{\bf $\mathbf{|\mathbf{E} \setminus F| > 1}.$\quad} 
    By A3, $E \setminus
    F$ is a proper subset in $\calF$. Thus, since $\calF$ is
    irreducible, there is a proper subset $H$ in $\calF$ such that
    $H\bowtie (E \setminus F)$.

    Suppose first that $F\subseteq H$, and consider the set $H' =
    H\cap E$.  Since $E$ intersects $F$, we have $H'\cap
    F\neq\emptyset$.  Further, $H'\cap (E\setminus F) = H\cap
    (E\setminus F)\neq\emptyset$, so $H'\setminus
    F\neq\emptyset$. Finally, $F\setminus E\neq\emptyset$ and
    $H'\subseteq E$, so $F\setminus H'\neq\emptyset$. Thus, $F\bowtie
    H'$.  However, $H'\setminus F= H\cap (E\setminus F)$ is a strict
    subset of $E\setminus F$, so $|H'\setminus F|<|E\setminus F|$, a
    contradiction with our choice of $E$.
 
    Thus, we have $F\not\subseteq H$. If, nevertheless, $F\cap
    H\neq\emptyset$, we set $H''=H\cap (E\cup F)$.  Clearly, we have
    $F\cap H''\neq\emptyset$.  Since $H''$ is a subset of $H$, we also
    have $F\setminus H''\neq\emptyset$.  Finally, since $H\cap
    (E\setminus F)\neq\emptyset$.  we have $H''\setminus
    F\neq\emptyset$. Thus, $H''\bowtie F$, yet $H''\setminus F=H\cap
    (E\setminus F)$ is a strict subset of $E\setminus F$, so
    $|H''\setminus F|<|E\setminus F|$, a contradiction with our choice
    of $E$.

    Hence, $H \intsct F = \emptyset$. However, this means that $E
    \setminus H$ still intersects $F$ nontrivially, and $|(E \setminus
    H) \setminus F| < |E \setminus F|$, a contradiction again.

  We have shown that assuming that $F \neq [m]$ leads to a
  contradiction.  Hence, $F = [m]$, which means that $\calF$ is a
  string of sausages.
\end{proof}

\newtheorem*{proembedding}{Proposition~\ref{pro:embedding}}

\begin{proembedding}
Let $\calC$ and $\calD$ be two clone structures over sets $C$ and $D$,
respectively, where $|C|=m$, $|D|=k$, and $C \cap D= \emptyset$.  Then
for each $c\in C$, the family of subsets $\calC(c \rightarrow \calD)$
is a clone structure.
\end{proembedding}
\begin{proof}%
  Fix a candidate $c\in C$, and let $\calR = (R_1, \ldots, R_n)$ and $\calQ =
  (Q_1, \ldots, Q_{n'})$ be two profiles of voters such that $\calC =
  \calC(\calR)$ and $\calD = \calC(\calQ)$. Since duplicating linear
  orders in $\calR$ and $\calQ$ does not change $\calC$ and $\calD$,
  we can assume without loss of generality that $n=n'\ge 2$.  Our goal
  is to construct a profile $\calR'$ such that $\calC(c \rightarrow
  \calD) = \calC(\calR')$.  This profile will have $n$ voters and
  $m+k-1$ alternatives, that is, $\calR' = (R'_1,\ldots, R'_n)$.  We
  will construct $\calR'$ in two steps. First, for each $i=1, \dots,
  n$, we set $R^0_i$ to be identical to $R_i$ except that the
  occurrence of $c$ is replaced by $Q_i$; denote the resulting profile
  by $\calR^0 = (R^0_1,\ldots, R^0_n)$ and let $\calC^0 =
  \calC(\calR^0)$.  It is easy to see that all elements of $\calC(c
  \rightarrow \calD)$ are clones in $\calR^0$, so $\calC(c \rightarrow
  \calD) \subseteq \calC^0$. If also $\calC^0\subseteq \calC(c
  \rightarrow \calD)$, we are done, since in this case we can set
  $\calR'=\calR^0$.

  Otherwise, we flip $Q_n$.  That is, assuming without loss of
  generality that $Q_n$ ranks the elements of $D$ as
  $$
  Q_n: d_1\succ d_2\succ\ldots \succ d_k
  $$
  and $R_n$ is given by $C_1\succ c \succ C_2$, we define
  $$
  R'_n: C_1\succ d_k\succ \ldots \succ d_1\succ C_2,  
  $$
  where we assume that $R'_n$ orders the elements of $C_1$ and $C_2$
  in the same way as $R_n$ does; we also set $R'_i=R^0_i$ for $i=1,
  \dots, n-1$.  Consider the resulting profile $\calR'$, and let
  $\calC'=\calC(\calR')$.  We claim that $\calC' = \calC(c \rightarrow
  \calD)$.  As above, it is easy to see that $\calC(c \rightarrow
  \calD) \subseteq \calC'$.  It remains to show that $\calC'\subseteq
  \calC(c \rightarrow \calD)$.

  Let $X$ be a ``parasite'' clone in $\calC^0\setminus \calC(c
  \rightarrow \calD)$. Clearly, it cannot be the case that $X\subseteq
  C$ or $X\subseteq D$. Further, if $D\subseteq X$, then $(X\setminus
  D)\cup\{c\}$ is a clone in $\calC$, and hence $X\in \calC(c
  \rightarrow \calD)$. Thus, the sets $C_X=X\cap C$ and $D_X=X\cap D$
  are both non-empty, and $D_X\neq D$.  By
  Proposition~\ref{pro:reversing}, we may assume that each order in
  $\calR^0$ is of the form
  $\ldots\succ C_X\succ D_X\succ D\setminus D_X\succ\ldots$.

  Now, suppose for the sake of contradiction that $Y$ is a clone in
  $\calC'\setminus \calC(c \rightarrow \calD)$.  By the same argument
  as in the previous paragraph, we conclude that $Y\cap
  D\neq\emptyset$, $Y\cap C\neq\emptyset$, and $D\not\subseteq
  Y$. Thus, we have two possibilities:
  \begin{itemize}%
  \item $d_1\in Y$, $d_k\not\in Y$. Then, since $Y$ is contiguous in
    $R'_1$ and $Y\cap C\neq\emptyset$, we have $Y\cap
    C_X\neq\emptyset$.  However, in $R'_n$ the element $d_k$ splits
    $Y$ and $C_X$, a contradiction.
  \item $d_1\not\in Y$, $d_k\in Y$. Then, since $Y$ is contiguous in
    $R'_n$ and $Y\cap C\neq\emptyset$, we have $Y\cap
    C_X\neq\emptyset$. However, in $R'_1$ the element $d_1$ splits $Y$
    and $C_X$, a contradiction.
  \end{itemize}
  Hence, we have $Y \in \calC(c \rightarrow \calD)$. The proof
  is complete.
\end{proof}

The above proof could be simplified if we were willing to use more
voters in the profile for $\calC(c \rightarrow \calD)$. However, the
current version of the proof is very useful when we consider the
number of voters needed to implement a particular clone structure.

\newtheorem*{thmcharmain}{Theorem~\ref{thm:char-main}}
\begin{thmcharmain}
  A family $\calF$ of subsets of $[m]$ is a clone structure if and only
  if it satisfies conditions A1--A5.
\end{thmcharmain}
\begin{proof}%
  We have already argued that any clone structure satisfies A1--A5; it
  remains to prove that the converse is also true.

  Our proof is by induction on $m$. Clearly the theorem holds for $m =
  1$ and for $m = 2$.  For the inductive step, assume it holds for
  each $m' < m$. Let $\calF$ be a family of subsets of $[m]$ that
  satisfies A1--A5.  If $\calF$ is irreducible then, by
  Theorem~\ref{thm:irreducible}, it is either a string of sausages or
  a fat sausage and thus a clone structure. Otherwise, $\calF$ contains
  a proper subfamily $\calD$.  Let $\calF'= \calF(\calD \rightarrow
  e)$ for some $e\notin [m]$.  We have argued that $\calD$ and
  $\calF'$ satisfy axioms A1--A5.  Hence, by our inductive hypothesis
  both $\calF'$ and $\calD$ are clone structures and so, by
  Proposition~\ref{pro:embedding}, $\calF = \calF'(e \rightarrow
  \calD)$ is a clone structure as well. This completes the proof.
\end{proof}

\subsection{Identifying Clone Structures in P}\label{app:poly}

\begin{theorem}\label{thm:poly}
  There exists a polynomial-time algorithm that, given a family of subsets
  $\calF$ over a finite set $F$, checks if $\calF$ is a clone structure
  over $F$.
\end{theorem}
\begin{proof}
  It is easy to see that we can check in polynomial time whether
  $\calF$ satisfies A1--A4.  Now, suppose that $\calF$ has passed this
  check, and it remains to verify that it satisfies A5.  We can
  directly check if $\calF$ contains a bicycle chain of size $3$, by
  considering all possible triples of the subsets in $\calF$. To check
  for bicycle chains of size $4$ or more, we will construct a directed
  graph $G$ as follows.

  The vertices of $G$ are ordered pairs $(X, Y)$, where $X$ and $Y$
  are two subsets in $\calF$ such that $X\bowtie Y$.  There is a
  directed edge from $(X, Y)$ to $(Y', Z)$ if $Y=Y'$, $X\cap Y\cap
  Z=\emptyset$.  and $Y\subseteq X\cup Z$.  Intuitively, $G$ has an
  edge from $(X, Y)$ to $(Y, Z)$ if $X$, $Y$ and $Z$ can be three
  consecutive sets in a bicycle chain.  It is not hard to verify that
  $G$ contains a directed cycle if and only if $\calF$ contains a
  bicycle chain of size $4$ or more.  Indeed, let $\{A_0, \dots,
  A_{k-1}\}$ be a bicycle chain of size $k\ge 4$ in $\calF$.  Then any
  pair $(A_i, A_{i+1})$ is a vertex of $G$. Moreover, there is an edge
  in $G$ between $(A_{i-1}, A_i)$ and $(A_i, A_{i+1})$, so $(A_0,
  A_1), (A_1, A_2), \dots, (A_{k-1}, A_0)$ is a directed cycle in $G$
  (as always in our discussion of bicycle chains, the indices are
  computed modulo $k$).  Conversely, if $G$ contains a directed cycle
  of the form $(X_0, X_1), (X_1, X_2), \dots, (X_{k-1}, X_0)$, then
  the sets $X_0, \dots, X_{k-1}$ form a bicycle chain.
\end{proof}

We can strengthen the proof of Theorem~\ref{thm:poly}
to obtain a logarithmic space algorithm.
Verifying axioms A1--A4 in logarithmic space is straightforward.
For axiom A5 the verification problem can be reduced
to connectivity testing for undirected
graphs; it then remains to
apply the breakthrough result of Reingold~\cite{rei:j:ustconinl}.

\section{Material Missing from Section~\ref{sec:representation}}
\label{app:representation}

We will now prove Theorem~\ref{thm:3rep}. However, to do so, we need 
the following proposition.

\newtheorem*{propirr3}{Proposition B.1}
\begin{proposition}\label{prop:irr3}
  Let $\calC$ be an irreducible clone structure over $[m]$.  If $\calC$
  is a string of sausages, it is $1$-implementable.  If $\calC$ is a
  fat sausage and $m>3$, then $\calC$ is $2$-implementable, but not
  $1$-implementable.  If $\calC$ is a fat sausage and $m=3$, then
  $\calC$ is $3$-implementable, but not $2$-implementable.
\end{proposition}
\begin{proof}%
  If $\calC$ is a string of sausages, it can be implemented using a
  single order, namely, $1\succ\ldots\succ m$.

  Now, suppose that $\calC$ is a fat sausage.  Clearly, it cannot be
  implemented with a single order, as the clone structure
  that corresponds to the latter is a string of sausages.

  Suppose first that $m=2k$. For convenience, set $x_i=i$, $y_i=k+i$
  for $i=1, \dots, k$.  We define $\calR=(R_1, R_2)$ as follows.
  \begin{align*}
    &R_1: x_1\succ \ldots \succ x_k\succ y_1\succ\ldots \succ y_k,\\
    &R_2: y_1\succ x_1\succ y_2\succ x_2\succ \ldots \succ y_k\succ
    x_k.
  \end{align*}
  We claim that $\calC=\calC(\calR)$.  Clearly, we have
  $\calC\subseteq\calC(\calR)$.  Now, suppose that
  $D\in\calC(\calR)\setminus\calC$, i.e., $|D|\neq 1, m$.  Since $D$
  has to be contiguous in $R_1$, we have one of the following three
  cases:
  \begin{itemize}%
  \item[(a)] $D=\{x_i, \dots, x_j\}$ for some $1\le i<j\le k$;
  \item[(b)] $D=\{y_i, \dots, y_j\}$ for some $1\le i<j\le k$;
  \item[(c)] $D=\{x_i, \dots, y_j\}$ for some $1\le i\le k$, $1\le
    j\le k$.
  \end{itemize}
  Case~(a) is impossible since in $R_2$ the element $y_j$ appears
  between $x_i$ and $x_j$. Similarly, case~(b) is impossible since in
  $R_2$ the element $x_i$ appears between $y_i$ and $y_j$.  In
  case~(c) we have $x_k, y_1\in D$. Since these elements appear at the
  opposite ends of $R_2$, we conclude that $D=[m]$, a contradiction.

  Next, suppose that $m=2k+1$, $k>1$. Set $x_i=i$, $y_i=k+i$ for $i=1,
  \dots, k$, $z=2k+1$.  We define $\calR=(R_1, R_2)$ as follows.
  \begin{align*}
    &R_1: x_1\succ \ldots \succ x_k\succ y_1\succ\ldots \succ y_{k-1}\succ z\succ y_{k},\\
    &R_2: y_1\succ x_1\succ y_2\succ x_2\succ\ldots\succ y_k\succ
    x_k\succ z.
  \end{align*}
  Again, it is clear that $\calC\subseteq\calC(\calR)$.  Now, suppose
  that $D\in\calC(\calR)\setminus\calC$, i.e., $|D|\neq 1, m$.  As in
  the case of even $m$, $D$ cannot be of the form $\{x_i, \dots,
  x_j\}$ for $1\le i<j\le k$, or of the form $\{y_i, \dots, y_j\}$ for
  $1\le i<j\le k-1$. Further, if $D$ is of the form $\{x_i, \dots,
  y_j\}$ for some $i=1, \dots, k$ and some $j=1, \dots, k-1$, then $D$
  must contain all elements that appear between $x_k$ and $y_1$ in
  $R_2$, i.e., either $D=[m]$ or $D=[m]\setminus\{z\}$, which is
  impossible.  Now, if $D$ contains $z$, it must also contain the only
  element that is adjacent to it in $R_2$, i.e., $x_k$.  As $y_1$
  appears between $x_k$ and $z$ in $R_1$, we have $y_1\in D$.  But
  then $D=[m]$, since $y_1$ and $z$ are extreme elements of $R_2$.

  Finally, if $m=3$, we can set $\calR=(R_1, R_2, R_3)$, where $R_1:
  1\succ 2 \succ 3$, $R_2: 2\succ 1 \succ 3$, $R_3: 2\succ 3\succ
  1$. To see that $\calC$ cannot be implemented by any $2$-voter
  profile $(R_1, R_2)$, observe that we can assume without loss of
  generality that $R_1$ is of the form $1\succ 2\succ 3$, and in $R_2$
  element $2$ is adjacent to at least one of the remaining elements
  (and hence forms a clone with that element).
\end{proof}

Now we are ready to prove Theorem~\ref{thm:3rep}. However, we will
first prove a stronger result, and only derive Theorem~\ref{thm:3rep}
as its corollary.

\newtheorem*{thm3repA}{Theorem B.2}
\begin{theorem}\label{thm:b:2}
  Any clone structure $\calC$ is $3$-implementable.  Moreover, if the
  tree $T(\calC)$ does not have nodes that carry labels of the form
  $x\odot y\odot z$, then $\calC$ is $2$-implementable.
  If $\calC$ is a string of sausages then it is $1$-implementable. 
\end{theorem}
\begin{proof}%
  If $\calC$ is a string of sausages, then it clearly is
  $1$-implementable. Otherwise, the following argument proves the theorem.

  Fix a clone structure $\calC$ on a set $C$ of size $m$.  If
  $T(\calC)$ does not have nodes that carry labels of the form $x
  \odot y \odot z$, then set $k=2$. Otherwise set $k=3$.  The proof is
  by induction on $m$. If $m=1$ or $m=2$, the theorem is obviously
  true. Further, if $\calC$ is irreducible, the theorem follows from
  Proposition~\ref{prop:irr3}.  Otherwise, $\calC$ contains a proper
  subfamily $\calD$.  By the inductive assumption, the clone
  structures $\calD$ and $\calC(\calD\rightarrow d)$, where $d\not\in
  C$, are $k$-implementable. Let $\calR=(R_1, \ldots, R_k)$ and
  $\calQ=(Q_1, \ldots, Q_k)$ be the respective preference profiles, i.e.,
  $\calC(\calD\rightarrow d)=\calC(\calR)$, $\calD =
  \calC(\calQ)$. Then the proof of Proposition~\ref{pro:embedding}
  shows how to combine $\calR$ and $\calQ$ to obtain a preference
  profile $\calR'$ with $k$ voters such that
  $\calC=\calC(\calR')$.
\end{proof}

Now the proof of Theorem~\ref{thm:3rep} is immediate.

\newtheorem*{thm3rep}{Theorem~\ref{thm:3rep}}
\begin{thm3rep}
Any clone structure is $3$-implementable.
\end{thm3rep}
\begin{proof}
  Follows directly from Theorem~\ref{thm:b:2}.
\end{proof}

\section{Material Missing from Section~\ref{sec:sp}}
\label{app:sp}

This part of the appendix contains the missing proofs and discussion
regarding clone structures in single-peaked elections.

Given a profile $\calR$ with
$|\peak(\calR)|\ge 2$ that is single-peaked with respect to $>$, we say that $p_1$
and $p_2$ are the {\em extreme peaks} of $\calR$ with respect to $>$
if either $p_1 > p > p_2$
for each $p \in \peak(\calR) \setminus \{p_1,p_2\}$ 
or $p_2 > p > p_1$ for each $p \in \peak(\calR) \setminus \{p_1, p_2\}$.  
We say that two candidates $a,
b \in C$ are {\em on the same side} of $c \in C\setminus\{a,b\}$ in
$>$ if either $(a > c \land b > c)$ or $(c > a \land c >
b)$. Otherwise, we say that $a$ and $b$ are {\em on the opposite
  sides} of $c$ in $>$.  Given two orders $>$ and $>'$ over $C$, we
say that $>$ and $>'$ {\em agree on $D \subseteq C$} if for each $a, b
\in D$ it holds that $a > b$ if and only if $a >' b$.

The single-peakedness of a given preference profile can be witnessed
by many different orders; 
for instance, if $\calR$ is single-peaked with respect to $>$, it is also 
single-peaked with respect to $\ola{>}$.  
However, it turns out that these orders have the same extreme
peaks and agree (up to an inversion) on all candidates between these
peaks.
 
\begin{theorem}
\label{thm:unique}
Consider a preference profile $\calR$ over $C$ with $|\peak(\calR)|\ge
2$ that is single-peaked with respect to two orders $>$ and $>'$. Let $p_1$ and
$p_2$ be the extreme peaks of $\calR$ with respect to $>$ such that
$p_1>p_2$. Then $p_1$ and $p_2$ are also the extreme peaks of $\calR$
with respect to $>'$. Moreover, if $p_1 >' p_2$, then $>$ and $>'$
agree on the set $P=\{c\mid p_1 > c > p_2\}\cup\{p_1, p_2\}$.
\end{theorem}
\begin{proof}
  Fix two orders $>$ and $>'$ that both witness the single-peakedness of $\calR$.
  Consider two candidates $p, q\in \peak(\calR)$,
  and another candidate $c\in C$.  We claim that either $p$ and $q$
  are on the same side of $c$ in both $>$ and $>'$, or they are on the
  opposite sides of $c$ in both $>$ and $>'$. Indeed, suppose that
  this is not the case.  Without loss of generality we assume that $p
  > q > c$ and $p >' c >' q$. Now, consider a preference order $R_i$
  such that $\peak(R_i) = p$. Since $\calR$ is single-peaked  with respect to 
  $>$, it must be the case that $p \succ_i q \succ_i c$; on the other hand, 
  since $\calR$ is single-peaked  with respect to $>$, we have 
  $p \succ_i c \succ_i q$, a contradiction.  
  Hence, either $p$ and $q$ are on the same side of
  $c$ in both $>$ and $>'$, or $p$ and $q$ are on the opposite sides
  of $c$ in both $>$ and $>'$.

  Now, consider an arbitrary $p\in\peak(\calR) \setminus \{p_1,
  p_2\}$. Since $p_1$ and $p_2$ are the extreme peaks of $\calR$ and $p_1 > p_2$, 
  we have $p_1 > p > p_2$. Therefore, by the
  argument above we have $p_1 >' p >' p_2$.
  This proves the first statement of the theorem.

  To prove the second statement, assume that $p_1 >' p_2$.  Also,
  without loss of generality, assume that $p_1=\peak(R_1)$ and
  $p_2=\peak(R_2)$.  Let $P'=\{c \mid p_1 >' c >' p_2\}\cup\{p_1,
  p_2\}$.  Suppose that $P\setminus P'\neq\emptyset$, and consider a
  candidate $c\in P\setminus P'$. The candidates $p_1$ and $p_2$ are
  on the opposite sides of $c$ in $>$, but on the same side of $c$ in
  $>'$, a contradiction.  Assuming $P'\setminus P\neq\emptyset$ leads
  to a contradiction as well. Thus, $P=P'$.  Now, suppose that for
  some $c, d\in P \setminus \{p_1, p_2\}$ we have $c > d$ and $d>' c$.
  Then, since $\calR$ is single-peaked  with respect to $>$, we have $p_1 \succ_1 c \succ_1 d$.
  However, since $\calR$ is single-peaked  with respect to $>'$, we have $p_1\succ_1 d\succ_1 c$, 
  a contradiction. Thus, the theorem is proved.
\end{proof}

Thus, by Theorem~\ref{thm:unique}, we can speak of extreme peaks of a single-peaked profile,
without referring to a specific societal axis.

\newtheorem*{thmsecondvariety}{Proposition~\ref{thm:second-variety}}

\begin{thmsecondvariety}
Let $\calR = (R_1, \ldots, R_n)$ be a preference profile over a
candidate set $C$ that is single-peaked  with respect to some order $>$, 
and let $D\in\calC(\calR)$ be a clone set for $\calR$ with $|D|\ge 2$.
Then $C$ can be partitioned
into pairwise disjoint sets $A_1$, $A_2$, $D_1$, $D_2$, and $P$ so that 
$C\setminus D=A_1\cup P\cup A_2$, $D=D_1\cup D_2$, 
$D_1\neq\emptyset$, $D_2\neq\emptyset$, and
$$
A_1 > D_1 > P > D_2 > A_2.
$$
Further, if $P\neq\emptyset$, then $\peak(\calR)\subseteq P$, and, moreover,  
$P \succ_i D \succ_i A_1 \union A_2$ for each $i=1, \dots, n$.
\end{thmsecondvariety}
\begin{proof}
  Suppose that our first claim is not true.  
  Then there exist two candidates $c_1, c_2\in
  C\setminus D$ and three candidates $d_1, d_2, d_3\in D$ such that
  $d_1 > c_1 > d_2 > c_2 > d_3$. Let $p=\peak(R_1)$.  If $p>d_2$ or
  $p=d_2$, we have $d_2\succ_1 c_2\succ_1 d_3$ and hence $D$ is not
  contiguous in $R_1$, a contradiction.  Similarly, if $d_2 > p$, we
  have $d_2\succ_1 c_1\succ_1 d_1$, a contradiction again. This proves
  our claim regarding the partition of $C$. To see that we can ensure 
  that both $D_1$ and $D_2$ are non-empty, note that if, e.g., 
  $D_1=\emptyset$, we can modify the partition by merging $P$
  into $A_1$ (so that the new $P$ is empty), and repartitioning $D$
  into two non-empty sets (recall that $|D|\ge 2$).

  To prove the second claim, consider an arbitrary preference
  profile $R_i$. Since $P \neq \emptyset$, the peak of $R_i$ 
  must be in $P$, as otherwise it would be impossible for $R_i$ to rank members
  of $D$ contiguously. 
  Now, let us show that $P \succ_i D$. Suppose that this is not the case, i.e., 
  $d\succ_i p$ for some $d\in D$, $p\in P$.  Since both $D_1$ and $D_2$ are non-empty, we can pick  
  two alternatives $d_1 \in D_1, d_2 \in D_2$.
  As $R_i$ ranks members of $D$ contiguously, it has to be the case that
  $D\succ_i p$, and, in particular, $\{d_1,d_2\} \succ_i p$.
  But we have $d_1>p>d_2$, so $d_1\succ_i p$ implies $p\succ_i d_2$, a contradiction.
  Thus $P \succ_i D$.

  Finally, let us show that $D \succ_i A_1 \cup A_2$. 
  Since the peak of $R_i$ is in $P$, $A_1 > D_1$ implies
  $D_1 \succ_i A_1$, and $D_2 >A_2$ implies $D_2 \succ_i A_2$. 
  Suppose for the sake of contradiction that $a\succ_i d$ for some $a\in A$, $d\in D$. 
  Then either $a\in A_1$, $d\in D_2$ (and hence $D_1\succ_i a\succ_i d$)
  or $a\in A_2$, $d\in D_1$ (and hence $D_2\succ_i a\succ_i d$). 
  In both cases, we obtain a contradiction with $D$ being contiguous in $R_i$.
  Thus, $D \succ_i A_1 \cup A_2$.  This completes the proof.
\end{proof}

\begin{proposition}
\label{thm:declone-to-sp}
Let $\calR = (R_1, \ldots, R_n)$ be a single-peaked preference profile
over a candidate set $C$, and let $D\in \calC(\calR)$ be a clone set
such that $|D| \geq 2$. Let $c$ be some candidate not in $C$. Then the
preference profile $\calR' = \calR(\calD \mapsto c)$ is single-peaked.
\end{proposition}
\begin{proof}
  Suppose that $\calR$ is single-peaked with respect to some order $>$ over $C$.
  We will now construct an order $>'$ over $(C\setminus D)\cup\{c\}$
  such that $\calR'$ is single-peaked with respect to $>'$.

  By Proposition~\ref{thm:second-variety}, we know that $C$ can be
  partitioned into $P$, $D_1$, $D_2$, $A_1$, $A_2$ so that $D_1 \neq
  \emptyset$, $D_2 \neq \emptyset$ and
  \[
    A_1 > D_1 > P > D_2 > A_2.
  \]
  If $P = \emptyset$, we set $>'$ to be an order that agrees
  with $>$ on $A_1 \cup A_2$ and satisfies $A_1 >' c >' A_2$.
  If $P \neq \emptyset$, we set $>'$ to be an order that agrees
  with $>$ on $A_1 \cup P \cup A_2$ and satisfies $A_1 >' c >' P >' A_2$.
  In both cases,
  it is immediate that $\calR'$ is single-peaked  with respect to $>'$.
\end{proof}

\newtheorem*{corgreedy}{Proposition~\ref{cor:greedy}} 
\begin{corgreedy}
  Let $\calR$ be a preference profile over a set of candidates $C$.
  Let $D^1, \ldots, D^k \in \calC(\calR)$ be a sequence of
  pairwise disjoint clone sets, and let $c^1, \ldots, c^k$ be
  a sequence of distinct candidates not in $C$. For each $i=1, \dots, k$,
  let $\calR_i$ denote a preference profile in which for
  each $j=1, \dots, k$, $j \neq i$, $D^j$ is decloned to $c^j$.
  Then  $\calR$ is single-peaked if and only if each of the profiles
  $\calR_i$, $1 \leq i \leq k$, is single-peaked.
\end{corgreedy}
\begin{proof}
  The ``only if'' direction follows from Proposition~\ref{thm:declone-to-sp}.
  For the ``if'' direction, 
  we will give a proof for the case $k=2$; the general case follows by induction.

  Set $\calR' = \calR_1(D^1 \mapsto c^1) =  \calR_2(D^2 \mapsto c^2)$, 
  and let $D=D^1$, $c=c^1$.
  By Proposition~\ref{thm:declone-to-sp},
  $\calR'$ is single-peaked; let $>'$ be some order that witnesses this.
  The general idea of our proof is as follows: We will first show that
  there are orders $>^1$ and $>^2$, both very similar to $>'$,
  witnessing single-peakedness of $\calR_1$ and $\calR_2$,
  respectively. We will then show that the ``edits'' needed to turn
  $>'$ into $>^1$ and the ``edits'' needed to turn $>'$ into $>^2$ are
  independent and thus we can turn $>'$ into an order witnessing
  single-peakedness of $\calR$.

  We will now construct an order $>^1$ that is very similar to $>'$
  and witnesses the single-peakedness of $\calR_1$.  We start with an
  arbitrary preference order $>$ that witnesses the single-peakedness
  of $\calR_1$.  This order may not have the properties that we are
  interested in; therefore, we will construct $>^1$ by taking a
  ``hybrid'' of $>$ and $>'$.
  We will consider two cases.
  \begin{description}
  \item[\label{case:type1}$\boldsymbol{D}$ is a clone set of the first type with respect to $\boldsymbol{>}$.] 
  We claim that in this case we can construct $>^1$
  from $>'$ by replacing $c$ with the members of $D$, ranked
  either according to $>$ or according to $\ola{>}$.

  Indeed, let $p_1$ and $p_2$ be the extreme peaks of $\calR'$ such that
  $p_1 >' p_2$, and let $B = \{b \mid p_1 >' b >' p_2\}$. If
  $|\peak(\calR')|=1$, set $p_1 = p_2$ to be the unique member
  of $\peak(\calR')$ and let $B = \{p_1\}$.
  
  Let $>^c$ be the order obtained from $>$ by replacing the
  occurrence of $D$ with $c$; if $p_2 >^c p_1$, reverse $>^c$. 
  The proof of
  Proposition~\ref{thm:declone-to-sp} shows that $\calR'$ is single-peaked 
  with respect to $>^c$. Thus, by Theorem~\ref{thm:unique}, $>^c$ and $>'$
  agree on $B$ and rank the members of $B$ consecutively.

  Suppose first that $c >^c B$. If $c >' B$ then
  we obtain $>^1$ from $>'$ by replacing the occurrence of $c$ with the
  members of $D$, ranked in the order of $>$.
  If $B >' c$, we replacing the occurrence of $c$ with the
  members of $D$, ranked in the order of $\ola{>}$.
  It is easy to see that $\calR_1$ is single-peaked with respect to $>^1$. 
  The case $B >^c c$ can be handled similarly.  
  Finally, suppose that $c \in B$. Since $>'$
  and $>^c$ agree on $B$ and rank members of $B$ consecutively, 
  it suffices again to replace the occurrence of $c$ with the
  members of $D$, ranked in the order of $>$.
  Clearly, $\calR_1$ is single-peaked with respect to the resulting order $>^1$.

  \item[\label{case:type2} $\boldsymbol{D}$ is a clone set of the second type
  with respect to $\boldsymbol{>}$.] 
  By Proposition~\ref{thm:second-variety}, there 
  is a (unique) partition of $C$ into sets $A_1, A_2, D_1, D_2, P$   
  such that $P \neq \emptyset$, $D_1\neq\emptyset$, $D_2\neq\emptyset$, and
  $A_1 > D_1 > P > D_2 > A_2$. Further, for each $R_i \in \calR_1$ it holds 
  that $P \succ_i D \succ_i A$. This implies that each $R'_i \in \calR'$ is of the form 
  $P \succ'_i c \succ'_i A$. Thus, both $P$ and $P\cup\{c\}$
  are clone sets for $\calR'$. Moreover, 
  $P\cap \peak(\calR')\neq\emptyset$, so by Corollary~\ref{cor:2type-nopeaks}  
  $P$ is a clone set of the first type with respect to $>'$. Thus, 
  we have either $A'_1 >' c >' P >' A'_2$ or $A'_1 >' P >' c >' A'_2$
  for some $A'_1, A'_2$ such that $A'_1\cup A'_2=A$; assume without loss of generality
  that  $A'_1 >' c >' P >' A'_2$. 
  Consider the order $>^1$ given by $A'_1 >^1 D_1 >^1 P >^1 D_2 >^1 A'_2$, 
  which agrees with $>'$ on $C \setminus D$ and agrees with $>$ on $D$.
  Clearly, $\calR_1$ is single-peaked with respect to $>^1$.
  \end{description}

In both cases, we derive an order $>^1$ that 
witnesses the single-peakedness of $\calR_1$ from $>'$
by either replacing the occurrence of $c$ with the members of $D$ 
(if $D$ is a clone set of the first type), or replacing $c$ with some members of $D$ and
inserting the remaining members of $D$ into a clearly specified position in
$>'$. Now, recall that $\calR_1(D^1\mapsto c^1)= \calR_2(D^2\mapsto c^2)$, 
and therefore $>'$ also witnesses the single-peakedness of $\calR_2(D^2\mapsto c^2)$.
Hence, we can derive an order $>^2$ that witnesses the single-peakedness of $\calR_2$
from $>'$ in a similar manner. 
Crucially, the ``edits'' required to obtain $>^1$ from $>'$ are
independent of the edits required to obtain $>^2$ from
$>'$. Consequently, if we apply these edits
jointly to obtain a new order $>^*$, this order
witnesses that $\calR$ is single-peaked, which is exactly what we need to prove.
\end{proof}

\newtheorem*{propirrsp}{Proposition~\ref{prop:irr-sp}}
\begin{propirrsp}
  Fat sausages and strings of sausages are single-peaked.
\end{propirrsp}
\begin{proof}
  It suffices to check that the profiles constructed in the proof of
  Proposition~\ref{prop:irr3} are single-peaked.
  For the string of sausages, this is immediate: the respective order
  $>$ is given by $1>\ldots>m$.
  For the fat sausage with $m=2k$, $k\ge 2$, we can use the order
  $x_k>\ldots>x_1>y_1>\ldots>y_k$.
  For the fat sausage with $m=2k+1$, $k\ge 2$, we can use the order
  $z>x_k>\ldots>x_1>y_1>\ldots>y_k$.
  Finally, for the fat sausage with $m=3$, we can set $1>2>3$.
\end{proof}

\newtheorem*{propsptree}{Proposition~\ref{prop:sp-tree}}
\begin{propsptree}
  Let $\calC$ be a clone structure over a set of candidates $C$. Suppose that
  for each Q-node of the PQ-tree decomposition $T(\calC)$ of $\calC$
  it holds that all children, except possibly the leftmost child and the
  rightmost child, are labeled with singletons, i.e., elements of $C$.
  Then $\calC$ is a single-peaked clone structure.
\end{propsptree}

\begin{proof}
  Let us fix $C$ and $\calC$ as in the statement of the theorem and
  let $T = T(\calC)$ be some PQ-tree decomposition of $\calC$. We
  first describe a societal axis $>$, and then construct a profile
  $\calR$ that is single-peaked with respect to $>$ and satisfies 
  $\calC(\calR) = \calC$.

  We obtain $>$ as follows. For every pair of candidates $c',c'' \in C$, we
  set $c' > c''$ if in the DFS traversal of $T$, the node representing
  $c'$ is visited before the node representing $c''$. Note that $T$
  is an ordered tree and thus the order of the DFS traversal is uniquely
  determined; intuitively, this is simply the left-to-right order of 
  leaves of $T$.

  To describe the profile $\calR$, we need some additional notation.
  As in Definition~\ref{def:d2c}, for every node $v \in T$, we set $C_v
  = \{c \in C \mid c \text{ is a leaf of $T$'s subtree rooted in } v\}$, 
  and let $D_v = C \setminus C_v$.  
  For every node $v \in T$, we will introduce several
  preference orders that rank the candidates in $C_v$ ahead of those in $D_v$. 
  This will ensure that the part of $\calC$ that corresponds to $v$ is
  implemented correctly. Since each preference order in $\calR$ has 
  to rank all candidates, we will define for each $v \in T$ an order
  $\succ_v$ on $D_v$ that will be used for ranking the candidates in $D_v$
  in the preference profiles that correspond to $v$.

  The order $\succ_v$ is defined as follows.
  Let $P=(v_1, \ldots, v_k)$ be the (unique) path from $v$ to the root, 
  where $v = v_1$, $v_k$ is the root, and for each $i=1, \dots, k-1$
  the node $v_{i+1}$ is the parent of $v_i$.
  Let $c, c'$ be two candidates in $D_v$.  
  Let $v_i$ be the first node of $P$ that lies on the path from $c$ to the root.  Similarly,
  let $v_j$ be the first node of $P$ that lies on the path from $c'$ to the root. 
  Then $\succ_v$ orders $c$ and $c'$ as follows:
  \begin{enumerate}
  \item If $i < j$ then $c \succ_v c'$.
  \item If $j > i$ then $c' \succ_v c$.
  \item If $i=j$ then:
    \begin{enumerate}
    \item If both $c$ and $c'$ are to the left of $C_{v_{i+1}}$ in $T$,
      then $c \succ_v c'$ if and only if $c' > c$.
    \item If both $c$ and $c'$ are to the right of $C_{v_{i+1}}$ in $T$,
      then $c \succ_v c'$ if and only if $c > c'$.
    \item If $c$ and $c'$ are on the opposite sides of $C_{v_{i+1}}$ in $T$,
      then the one to the left of $C_{v_{i+1}}$ precedes in $\succ_v$ the
      one to the right of $C_{v_{i+1}}$.
    \end{enumerate}
  \end{enumerate}
  Let $\succ$ be some preference order over $C_v$ that is single-peaked with   
  respect to $>$.
  The reader can verify that any preference order that ranks $C_v$ above $D_v$, 
  agrees with $\succ$ on $C_v$, and agrees with $\succ_v$
  on $D_v$ is single-peaked with respect to $>$. 
  Further, for each clone $D \in \calC$ such that
  either $C_v \subseteq D$ or $C_v \cap D = \emptyset$ it holds that
  the members of $D$ are ranked consecutively in $\succ_v$. 
  (The last claim uses the fact that for any Q-node 
   only its leftmost and rightmost child can be non-singletons.)  

  We can now describe the profile $\calR$. For each node $v$ we
  construct several preference orders as follows.

  Let $v$ be a P-node with children $v_1, \ldots, v_k$. For
  each $i=1, \dots, k$, we add four preference orders, 
  which we will denote by $R_{v_i}^1$, $R_{v_i}^2$,
  $R_{v_i}^3$, and $R_{v_i}^4$. Each of them ranks $C_v$ 
  above $D_v$ and agrees with $\succ_v$ on $D_v$. Thus, it
  remains to describe how they order the members of $C_v$.
  Set $A = C_{v_1} \cup \cdots \cup C_{v_{i-1}}$, 
  $B = C_{v_{i+1}} \cup \cdots \cup C_{v_k}$. 
  We have (see the description below for clarification):
  \begin{align*}
    R^1_{v_i}&:   C_{v_i} \succ A \succ B, \\
    R^2_{v_i}&:   \revnot{C_{v_i}} \succ A \succ B, \\
    R^3_{v_i}&:   C_{v_i} \succ B \succ A, \\
    R^4_{v_i}&:   \revnot{C_{v_i}} \succ B \succ A.
  \end{align*}
  For each occurrence of $A$ and $B$ in the above preference orders,
  we order member of $A$ and $B$ either following $>$ or the reverse
  of $>$, whichever way is required to ensure
  single-peakedness. Each occurrence of $C_{v_i}$ corresponds to ranking 
  the members of $C_{v_i}$ according to $>$, and each occurrence of
  $\revnot{C_{v_i}}$ corresponds to ranking the members of $C_{v_i}$ 
  according to  $>$.

  The preference orders $R^j_{v_i}$, $j=1, \dots, 4$, ensure
  that any clone in $\calR$ that contains both a member of $C_{v_i}$ and
  a member of $C_{v}\setminus C_{v_i}$ has to contain all of $C_v$. 
  Further, for each clone $D \in \calC$, the members of $D$ are ranked
  consecutively in each $R^j_{v_i}$, $j=1, \dots, 4$.

  Now, let $v$ be a Q-node in $T$ with children $v_1, \ldots, v_k$;
  note that $k \geq 3$, since in our construction of PQ-trees
  all nodes with 2 children are labeled as P-nodes.
  Then we introduce voters $R_{v_1}^1$, $R_{v_1}^2$,
  $R_{v_1}^3$, and $R_{v_1}^4$, and $R_{v_k}^1$, $R_{v_k}^2$,
  $R_{v_k}^3$, and $R_{v_k}^4$, defined in the same way as for a P-node. 
  This completes the description of $\calR$
  
  Clearly, each set in $\calC$ is a clone in the preference profile $\calR$.
  Conversely, no set $D\in 2^C\setminus\calC$  is a clone in $\calR$. 
  Indeed, fix a subset $D\in 2^C\setminus\calC$, 
  let $c$ be the minimal element of $D$ with respect to $>$, 
  and let $c'$ be the maximal element of $D$ with respect to $>$. Now consider
  the path $v_1, \ldots, v_t$ from $c$ to $c'$ in $T$. 
  By construction of $\calR$ we have
  $\bigcup_{i=1}^t C_{v_i} \subseteq D$. However, by the choice of $c$
  and $c'$ we have $D = \bigcup_{i=1}^t C_{v_i} \in \calC$, which is a
  contradiction.  
\end{proof}

\section{Complete Decloning Algorithm}\label{app:declone-alg}

In this section we give a polynomial-time algorithm that finds optimal
decloning toward a single-peaked profile without making any
assumptions on the structure of the initial profile.

We start by proving a preliminary lemma, which may be of independent
interest.

\begin{lemma}\label{prop:type1-strings}
  Let $\calR$ be a single-peaked preference profile over a candidate set $C$,
  and let $D\in\calC(\calR)$ be a clone set such that the set family
  $\calD = \{X\in\calC(\calR)\mid X\subseteq D\}$ is a string of sausages.   
  If $|D| \geq 3$, then 
  $\calR$ is single-peaked with respect to some order $>$ such that
  $D$ is a clone set of the first type with respect to $>$.
\end{lemma}
\begin{proof}
  Suppose that $|D|\ge 3$ and  $\calR$ is single-peaked with respect to some order $>'$
  such that $D$ is a clone set of the second type with respect to $>'$.
  By Proposition~\ref{thm:second-variety},
  there is a partition of $C$ into sets $P$, $D_1$, $D_2$, $A_1$, $A_2$
  such that $D = D_1 \cup D_2$, $D_1 \neq \emptyset$, $D_2 \neq \emptyset$,  
  and $A_1 >' D_1 >' P >' D_2 >' A_2$.
  By the same proposition, $\peak(\calR) \subseteq P$.

  Assume without loss of generality that the first voter in $\calR$
  ranks the candidates in $D$ as $d_1\succ_1\dots\succ_1 d_m$.
  We will now argue that each voter in $\calR$
  ranks the candidates in $D$ according to $\succ_1$.
  Since $\calD$ is a string of sausages, each voter's
  ranking of the candidates in $D$ coincides with either $\succ_1$ or $\ola{\succ_1}$.
  Now, suppose that some voter $i$ ranks the candidates in $D$
  according to $\ola{\succ_1}$. Since $|D| \geq 3$, at least
  one of the sets $D_1, D_2$ has at least two elements.
  Assume without loss of generality that $|D_1|\ge 2$ and let
  $d, d'$ be two candidates in $D_1$ such that and $d \succ_1 d'$
  (and hence $d' \succ_i d$).
  Since $\peak(R_1)\in P$ and $D_1>P$, we have $d' >' d$.  
  However, since $\peak(R_i)\in P$, it has to be the case
  that $d \succ_i d'$,  a contradiction.
  
  This implies that $\calR$ is single-peaked with respect to
  an order $>$ that agrees with $>'$ on $C \setminus D$, agrees
  with $\succ_1$ on $D$, and is of the form
  $A_1 > P > D > A_2$.
  Clearly, $D$ is a clone of the first type with respect to $>$.
\end{proof} 

We will now use Lemma~\ref{prop:type1-strings} to show that it is never beneficial to 
partially collapse a string of sausages into a clone of size $3$ or more.

\begin{proposition}\label{cor:strings-all-or-nothing}
 Let $\calR$ be a preference profile over a set of candidates $C$,
 and let $D\in\calC(\calR)$ be a clone set such that
 the set family $\calD = \{X\in\calC(\calR)\mid X\subseteq D\}$ is a
 string of sausages.  Let $D' \in \calD$ be a clone set, $D' \neq D$,
 $|D'| \geq 2$, and let $d'$ be a candidate not in $C$.  If
 $\calR' = \calR(D' \mapsto d')$ is single-peaked and
 $|D\setminus D'\cup\{d'\}|\ge 3$,
 then $\calR$ is single-peaked as well.
\end{proposition}
\begin{proof}
 Let $D=\{d_1, \dots, d_m\}$. We can assume that the first voter
 ranks the candidates in $D$ as $d_1\succ_1\ldots\succ_1 d_m$,
 and $D'=\{d_i, \dots, d_j\}$, where $1\le i<j\le m$.
 Note that all voters in $\calR$ rank the candidates in $D$
 either according to $\succ_1$ or according to $\ola{\succ_1}$, and therefore
 $D\setminus\{d_1, d_m\}$ does not contain any peaks of $\calR$.
 Set $D^*=D\setminus D'\cup\{d'\}$.
 Clearly,  $D^*$ is a string of sausages in $\calR'$.  Thus, by Lemma~\ref{prop:type1-strings} the 
profile $\calR'$
 is single-peaked with respect to some order $>'$ such that
 $D^*$ is a clone set of the first type with respect to $>'$,
 i.e., we have $A_1 >' D^* >' A_2$,  where $A_1\cup A_2=C\setminus D$.

 Suppose that $\peak(\calR')\not\subseteq D^*$.
 Then, since $D^*$ is a string of sausages in $\calR'$,
 the restriction of $>'$ on $D^*$ is either of the form
 $$
 d_1 >' \ldots >' d_{i-1}>' d' >' d_{j+1}>' \ldots >' d_m
 $$
 or of the form
 $$
 d_m >' \ldots >' d_{j+1}>' d' >' d_{i-1}>' \ldots >' d_1;
 $$
 by reversing $>'$ if necessary, we can assume the former.
 Consider the order $>$ on $C$ that agrees with $>'$ outside
 of $D$ and ranks the elements of $D$ as
 $d_1>\ldots >d_m$. It is easy to see that $\calR$
 is single-peaked with respect to $>$.

 On the other hand, suppose that $\peak(\calR')\subseteq D^*$.  This
 means that $\peak(\calR) \subseteq \{d_1, d_m\}$. If
 $|\peak(\calR')| = 2$ then $>'$ is of the form as in the
 paragraph above and we can obtain $>$ from it the same way.
 If $|\peak(\calR')| = 1$ then all preference orders
 in $\calR'$ rank voters in $D^*$ identically. We obtain $>$ from
 $>'$ by replacing the occurrence of $D^*$ with the occurrence
 of $D$, orders from $d_1$ to $d_m$.
\end{proof}

We will now briefly describe our algorithm {\sc DecloneSP}$(\calR)$
that finds an optimal decloning of $\calR$ towards a single-peaked
profile.  It proceeds as {\sc
  BasicDecloneSP}$(\calR)$, with one exception.  Specifically, when
{\sc DecloneSP}$(\calR)$ processes a node $v$ of type Q with $m$
children $v_1, \dots, v_m$ and discovers that $v$ cannot be colored
white, it considers two ways of splitting its children into two
contiguous groups, namely, $(\{v_1\}, \{v_2, \dots, v_m\})$ and
$(\{v_1, \dots, v_{m-1}\}, \{v_m\})$.  For each split, it declones the
non-singleton group, i.e., removes the respective branches of the tree
and replaces them with one black node.  It then checks if the
resulting profile is single-peaked.  If yes, it might be possible to
recursively expand the singleton node, so {\sc DecloneSP}$(\calR)$
calls itself recursively on the respective subtree.  After the
recursive calls that correspond to both splits return, the algorithm
chooses the better split.  We postpone the formal description of the
algorithm and its proof of correctness to the full version of the
paper; in essence, Proposition~\ref{cor:strings-all-or-nothing} shows
that is suffices to consider the splits of the form described above,
and an argument similar to that in the proof of
Theorem~\ref{thm:declone-alg} shows that different branches of the
tree can be handled independently.

\section{Material Missing from  Section~\ref{sec:sc}}
\label{app:sc}

\newtheorem*{propsceasy}{Proposition~\ref{prop:sc-easy}}
\begin{propsceasy}
The problem of checking if a given profile is single-crossing is in $\p$.
\end{propsceasy}
\begin{proof}
Suppose that we are given a preference profile $\calR=(R_1, \dots, R_n)$ over a set
of candidates $C=\{c_1, \dots, c_m\}$. 
For each $i=1, \dots, n$, we will check if there
exists an order $\lhd$ over $[n]$ in which voter $i$ appears first;
we return ``yes'' if the answer is positive for some $i$.

Without loss of generality, we can focus on the case $i=1$ and assume
that voter $1$ ranks the candidates as $c_1\succ_1 \ldots \succ_1 c_m$.
Now, consider a directed graph $G$ with a vertex set $[n]$ that has
an edge from $j$ to $k$ if there exists a pair of candidates $(c_x, c_y)$
with $x<y$ such that $c_x\succ_j c_y$, but $c_y\succ_k c_x$.
Clearly, for every pair of nodes $(j, k)$ at least one of the edges
$(j, k)$ and $(k, j)$ is present in this graph. The single-crossing condition
immediately implies that if there is an edge from $j$ to $k$, 
then $j$ precedes $k$ in $\lhd$. Therefore, if $G$ has a cycle, 
$\calR$ is not single-crossing with respect to any ordering of voters
in which voter $1$ appears first. On the other hand, if $G$ is
acyclic, it induces a total order over $[n]$, and it is immediate
that $\calR$ is single-crossing with respect to this order.  
\end{proof}

We will now construct the tools needed to prove
Theorems~\ref{thm:scimp} and~\ref{thm:scnp}.  The following is an
immediate but useful corollary to the definition of single-crossing
profiles; to some degree, it justifies viewing single-crossing
profiles as collections of voters over some spectrum of views.

\begin{observation}\label{cor:sc-reverse}
  If profile $\calR = (R_1, \ldots, R_n)$ is single-crossing with
  respect to the order $1 \lhd 2 \lhd \cdots \lhd n$
  then $\calR' = (R_1, \ldots, R_n, \revnot{R_1})$ is single-crossing
  as well.
\end{observation}

We now define a family of single-crossing profiles that 
implements a fat sausage; we will also use it within the reduction in
Theorem~\ref{thm:scnp}.

\begin{definition}\label{def:slide}
  Let $m$ be a positive integer, $m > 2$, and let $C$ be an
  $m$-element candidate set. Rename the candidates so that $C = \{a,
  b, c_1, \ldots, c_{m-2}\}$. We call each profile $\calR = \{R_1,
  \ldots, R_m\}$ of the form:
  \begin{align*}
    R_1: & a \succ_1 b \succ_1 c_2 \succ_1\cdots \succ_1 c_{m-2} \\
    R_2: & b \succ_2 a  \succ_2 c_2 \succ_2 \cdots \succ_1 c_{m-2} \\
    R_3: & b \succ_3 c_2 \succ_3 a \succ_3 \cdots \succ_1 c_{m-2} \\
    &\vdots \\
    R_m: & b \succ_m c_2 \succ_m c_3 \succ_m \cdots \succ_m c_{m-2} \succ_m a 
  \end{align*}
  a {\em slide} over $C$. We refer to $a$ as the {\em sliding candidate}.
\end{definition}

\begin{proposition}\label{prop:sc-fat}
  Let $C$ be a set of candidates and let $\calR$ be a slide over $C$.
  $\calR$ is single-crossing and $\calC(\calR)$ is a fat sausage.
\end{proposition}
\begin{proof}
  Let $\calR = (R_1, \ldots, R_n)$ be a slide over $C$ (we will use
  the notation as in Definition~\ref{def:slide}). $\calR$ is clearly
  single-crossing with respect to the order $1\lhd \cdots\lhd m$. 
  We claim that $\calC(\calR)$ is a fat sausage. To see this, first note
  that if $D$ is a clone set and $a, b \in D$ then $D = C$. This is
  so, because $R_m$ ranks $b$ first and $a$ last.  Let $D$ be a clone
  set in $\calC(\calR)$. It is easy to verify that if $D$ contains
  more than one candidate then it must contain $a$ and $b$, and thus
  $D = C$. This completes the proof.
\end{proof}

\begin{proposition}\label{prop:slide-unique}
  Let $m$ be a positive integer, $m > 2$, and let $\calR =
  (R_1, \ldots, R_m)$ be a slide over $m$-candidate set $C$. $\calR$
  is not single-crossing with respect to orders other than  
  $1 \lhd_1 2 \lhd_1 \cdots \lhd_1 m$ and $m \lhd_2 m-1 \lhd_2 \cdots \lhd_2 1$.
\end{proposition}
\begin{proof}
  For the sake of contradiction assume that there is an order
  $\lhd$ different from $\lhd_1$ and $\lhd_2$ with respect to which
  $\calR$ is single-crossing. Suppose first that there
  are some $i, j, k \in [m]$ such that $i < j < k$ and $i \lhd k \lhd j$.

  Let $a$ be the sliding candidate in $\calR$.  Since $j \neq m$,
  there is some candidate $c$ such that both $R_i$ and $R_j$ rank
  $a$ above $c$, but $R_k$ ranks $c$ above $a$. However, this immediately
  implies that $\calR$ is not single-crossing with respect to
  $\lhd$.  The remaining three cases (in which there exist $i <
  j < k$ such that $j \lhd i \lhd k$, $j \lhd k \lhd i$, 
  or $k \lhd i \lhd j$) can be handled similarly.
\end{proof}

\newtheorem*{thmscimp}{Theorem~\ref{thm:scimp}}
\begin{thmscimp}
For every clone structure $\calC$ there exists a single-crossing profile
$\calR$ such that $\calC = \calC(\calR)$.
\end{thmscimp}
\begin{proof}
  It is easy to see that each string of sausages can be implemented
  with a single-crossing profile because a profile with a single voter
  suffices. Fat sausages can be implemented by
  Proposition~\ref{prop:sc-fat}. These are the only two types of
  irreducible clone structures and, thus, to prove the theorem it
  remains to show that clone structures implementable using
  single-crossing profiles are closed under embeddings.

  Let $\calC$ and $\calD$ be two clone structures over disjoint sets
  $C$ and $D$, such that both $\calC$ and $\calD$ can be implemented
  with a single-crossing profile.
  Let $\calR = (R_1, \ldots, R_n)$ be a single-crossing profile over
  $C$ such that $\calC = \calC(\calR)$ and let $\calQ = (Q_1, \ldots,
  Q_n)$ be a single-crossing profile over $D$ such that $\calD =
  \calC(Q)$. Note that since we can freely duplicate preference orders
  in a single-crossing profile, we can assume that both $\calR$ and
  $\calQ$ have the same number of voters.  Further, by
  Corollary~\ref{cor:sc-reverse}, we can assume that $Q_1$ and $Q_n$
  are reverses of each other, and we can assume that both $\calR$ and
  $\calQ$ are single-crossing with respect to the standard order over
  integers ($1 \lhd 2 \lhd \cdots \lhd n$).

  Fix an arbitrary candidate $c \in C$. We will construct a profile
  $\calP = (P_1, \ldots, P_{2n-1})$ such that $\calC(c \rightarrow
  \calD) = \calC(\calP)$. For each $i = 1, \ldots, n$, we define $P_i$
  to be identical to $R_i$, except that we replace the occurrence of
  $c$ with $Q_1$. For each $k = 2, \ldots, n$, we define $P_{(n-1)+k}$
  to be identical to $R_n$, except that we replace the occurrence of
  $c$ with $Q_k$.  (Somewhat abusing our notation, we could write
  $\calP = (R_1(c \rightarrow Q_1), R_2(c \rightarrow Q_1), \ldots,
  R_n(c \rightarrow Q_1), R_n(c \rightarrow Q_2), \ldots, R_n(c
  \rightarrow Q_n))$.)

  It is easy to verify that since both $\calR$ and $\calQ$ are
  single-crossing then so is $\calP$. Similarly, it is easy to verify
  that $\calD \subseteq \calC(\calP)$ and that for each $A \in \calC$,
  if $c \notin A$, then $A \in \calC(\calP)$ and if $c \in A$ then $(A
  \setminus \{c\})\cup D \in \calC(\calP)$. Thus, it remains to show
  that $\calP$ does not contain any ``parasite'' clones.

  Clearly, any clone set in $\calC(\calP)$ that contains members
  of $C$ only or members of $D$ only belongs to $\calC(c \rightarrow
  \calD)$. Consider preference orders $P_n$ and $P_{2n-1}$. By
  construction, $P_n$ is $R_n$ with $c$ replaced by $Q_1$ and
  $P_{2n-1}$ is $R_n$ with $c$ replaced by $\revnot{Q_1}$ (recall the
  assumption that $Q_n = \revnot{Q_1}$). Thus, any clone set of
  $\calC(\calP)$ that contains both a member of $C$ and a member of
  $D$ must contain all members of $D$.  By construction, any clone set
  $X \in \calC(\calP)$ that contains all members of $D$ belongs to
  $\calC(c \rightarrow \calD)$. This completes the proof.
\end{proof}

The above embedding construction stands in a sharp contrast to the
construction used for unrestricted profiles
(Proposition~\ref{pro:embedding}).  There, we could embed one clone
structure into the other without increasing the number of voters
needed to implement the profiles. Here, every embedding operation
nearly doubles the number of votes. Thus, for unrestricted preferences
we need between $1$ and $3$ voters to implement any given clone
structure,
but for single-crossing profiles our construction may require a number of
voters that is exponential in $|C|$.  On the other
hand, our construction is certainly not optimal (to see this, compare
the clone structure built in the proof of Theorem~\ref{thm:scnp} and
the profile that implements it there, and the profile that would arise
from applying the construction from Theorem~\ref{thm:scimp} to
implement this clone structure).  It is interesting to ask if it is possible to
implement every clone structure over candidate set $C$ with a
single-crossing profile with at most $\poly(|C|)$ preference orders.

Let us now turn to the issue of decloning toward a single-crossing
profile. Unfortunately, as opposed to the case of single-peaked
elections, for single-crossing the problem is $\np$-complete.
Our reduction uses the standard $\np$-complete problem
Exact Cover by $3$-Sets (X3C). %

\begin{definition}
  An instance $I = (B,\calS)$ consists of a base set $B = \{b_1,
  \ldots, b_{3k}\}$ and a collection $\calS = \{S_1, \ldots, S_s\}$ of
  $3$-element subsets of $B$.  It is a \emph{yes}-instance if there
  exists a set $A \subseteq \{1, \ldots, s\}$ such that (a)
  $\bigcup_{i \in A}S_i = B$ and (b) for each $i, j \in A$, $S_i \cap
  S_j = \emptyset$.
\end{definition}

\newtheorem*{thmscnp}{Theorem~\ref{thm:scnp}}
\begin{thmscnp}
  Given a profile $\calR$ over a candidate set $C$ and a positive
  integer $k$, it is $\np$-complete to decide if there exists a
  single-crossing profile $\calR'$ with $c(\calR') \geq k$ such that 
  $\calR'$ can be obtained from $\calR$ by decloning.
\end{thmscnp}
\begin{proof}
  Proposition~\ref{prop:sc-easy} implies that this problem is in NP:
  it suffices to guess clone sets that need to be decloned
  and use the algorithm from the proof of Proposition~\ref{prop:sc-easy}
  to verify that the resulting preference profile is single-crossing.
  
  To prove that this problem is NP-hard, 
  we give a reduction from X3C. Let $I = (B,\calS)$ be an input
  instance of X3C with $B = \{b_1, \ldots, b_{3k}\}$ and $\calS =
  \{S_1, \ldots, S_s\}$. By duplicating sets in $\calS$ 
  if necessary, we can assume that $s > 3k$.

  We construct a profile $\calR$ in stages. First let $\calP = (P_1,
  \ldots, P_s)$ be a slide over $[s]$; to be specific, we pick $1$ to
  be the sliding candidate, but this choice is irrelevant for the
  proof. For each $i = 1, \ldots, s$, we set $\calP'_i$
  to be a group of $2s$ preference orders, denoted $P'_{i.1}, \ldots,
  P'_{i,t}$, each identical to $P_i$; further in our construction we
  will modify the members of each group appropriately. We define
  $\calP'$ to consist exactly of these $s$ groups of orders: that is,
  abusing notation, $\calP' = \calP'_1 + \cdots +\calP'_s$.
  Intuitively, for $i = 1, \ldots, 3k$, group $\calP'_i$ corresponds to
  the element $b_i$ of $B$ and, after further modifications that we
  will introduce in the profile, the role of group $\calP'_i$ will be to
  ensure that it is impossible to pick two sets from $\calS$ that
  both contain $b_i$. The
  remaining $s-3k$ groups are added for the sake of uniformity and to
  maintain the slide structure within $\calP'$.

  Let $\{C_1, \ldots, C_s\}$ be a family of disjoint candidate sets
  such that $|C_j|=6s$ for each $j=1, \dots, s$.
  For each $j = 1, \ldots, s$, we let $\calQ'^j =
  (Q'^j_1, \ldots, Q'^j_{6s})$ be a slide over $C_j$ (picking
  an arbitrary member of $C_j$ to be the sliding candidate). 
  We obtain profile $\calQ^j$ by splitting $\calQ'^j$ into three
  contiguous groups of size $2s$ each and swapping the $2j$-th and $(2j-1)$-th voter
  in each group (that is, $2j$ and $2j-1$, $2j+2s$ and $2j-1+2s$, and $2j+4s$ and $2j-1+4s$).
  Clearly, $\calQ^j$ is single-crossing
  with respect to the voter ordering $\lhd_j$ obtained from the standard
  ordering $1 \lhd \cdots \lhd 6s$ by swapping the same pairs of voters, 
  namely, $2j$ and $2j-1$, $2j+2s$ and $2j-1+2s$, and $2j+4s$ and $2j-1+4s$.
  Moreover, by Proposition~\ref{prop:slide-unique}
  the profile $\calQ^j$ is single-crossing with respect to $\lhd_j$ and its
  reverse, but not with respect to any other order. 

  We construct our final profile $\calR$ from $\calP'$ by embedding
  the profiles $\calQ^j$, $j=1, \dots, 6s$, into it. Specifically, 
  for each $S_j \in \calS$ such that $S_j = \{b_x,b_y,b_z\}$, $x < y < z$, we replace
  candidate $j$ with one of the preference orders from $\calQ^j = (Q^j_1,
  \ldots, Q^j_{6s})$ as follows: 
  \smallskip
  \begin{enumerate}
  \item For each $i$, $1 \leq i < x$, in group $\calP'_i$ we replace
    candidate $j$ with preference order $Q^j_1$.

  \item For $i = x$ and for each $\ell = 1, \ldots, 2s$, we replace
    candidate $j$ in $P'_{i,\ell}$ with $Q^j_\ell$.

  \item For each $i$, $x < i < y$, in group $\calP'_i$ we replace
    candidate $j$ with preference order $Q^j_{2s}$.

  \item For $i = y$ and for each $\ell = 1, \ldots, 2s$, we replace
    candidate $j$ in $P'_{i,\ell}$ with $Q^j_{2s+\ell}$.

  \item For each $i$, $y < i < z$, in group $\calP'_i$ we replace
    candidate $j$ with preference order $Q^j_{4s}$.

  \item For $i = z$ and for each $\ell = 1, \ldots, t$, we replace
    candidate $j$ in $P'_{i,\ell}$ with $Q^j_{4s+\ell}$.

  \item For each $i$, $z < i < s$, in group $\calP'_i$ we replace
    candidate $j$ with preference order $Q^j_{6s}$.
  \end{enumerate}

  As a result, we obtain a profile $\calR$ over the 
  candidate set $C = \bigcup_{i=1}^sC_i$, consisting of $s$ groups of voters, 
  $\calR_1, \ldots, \calR_s$, where each group $\calR_i$ contains $2s$ voters denoted
  by $R_{i,1}, \ldots, R_{i,2s}$. 
  Clearly, $\calR$ can be constructed in time polynomial in $s$.
  We claim that $I$ is a \emph{yes}-instance of X3C if and only if
  there exists a single-crossing profile $\calT$ with $c(\calT) \geq 2sk+(s-k)$
  that can be obtained by decloning $\calR$.

  Let us assume that there exists a single-crossing profile $\calT$ 
  with $c(\calT) \geq 2sk+(s-k)$ that can be
  obtained by decloning $\calR$. Just as $\calR$ was divided into $s$
  groups $\calR_1, \ldots, \calR_s$, $\calT$ is divided into
  corresponding $s$ groups $\calT_1, \ldots, \calT_s$. We will show
  that in this case $I$ is a \emph{yes}-instance of X3C.

  Note that $\calC(\calR)$ is a composition of a fat sausage over $[s]$
  and fat sausages over the candidate sets $C_i$ for $i=1, \dots, s$. Thus, $C_1, \ldots,
  C_s$ are the only nontrivial clones in $\calC(\calR)$.  
  Define $A = \{ i \mid C_i \text{ is not decloned in } \calT \}$. We claim that
  $A$ corresponds to an exact cover of $B$, i.e., 
  $\bigcup_{i \in A}S_i = B$ and the sets $S_i$, $i\in A$, 
  are pairwise disjoint.

  Since $c(\calT) \geq 2sk+(s-k)$ and $C_i$s are the only nontrivial
  clones in $\calT$, it must be the case that $|A| \geq k$. Now, consider two
  sets $S_j$ and $S_k$ such that $j,k \in A$. Assume that there exists
  an element $b_i$ such that $b_i \in S_j \cap S_k$.  We will
  show that this implies that $\calT$ is not single-crossing.  
  Indeed, since $C_j$ is not decloned and $b_i \in S_j$, any
  order of voters witnessing single-crossingness of $\calT$ has to
  order the voters in $\calT_i = (T_{i,1}, \ldots, T_{i,t})$ according
  to $\lhd_j$ or its reverse.
  Similarly, since $b_i \in S_k$ and $C_k$ has not been decloned,
  any order witnessing single-crossingness of $\calT$ has to order the
  voters in $\calT_i$ according to $\lhd_k$ or
  its reverse.  By construction, $\lhd_j$ and $\lhd_k$ are neither identical
  nor each others' reverses. Thus, it must be the case that $S_j \cap S_k =
  \emptyset$. In consequence, it must be the case that $|A| = k$ and
  $\bigcup_{i \in A}S_i = B$.
  
  For the other direction, let us assume that $I$ is a
  \emph{yes}-instance and let $A \subseteq \{1, \ldots, s\}$ be such
  that $\bigcup_{i \in A}S_i = B$ and for each $i,j \in A$ we
  have $S_i \cap S_j = \emptyset$. Let $\calT$ be a profile obtained
  from $\calR$ be decloning each set $C_i$ such that $i \notin A$.  We
  claim that $\calT$ is single-crossing. To see this, it suffices to
  take an order $\lhd$ that orders the groups $\calT_1, \ldots,
  \calT_n$ of $\calT$ as $\calT_1 \lhd \calT_2 \lhd \cdots \lhd \calT_s$,
  within each group $\calT_i$, $i=1, \dots, 3k$, orders the voters
  according to $\lhd_j$, where $j \in A$ and $b_i \in S_j$
  (by choice of $A$ such a $j$ is unique), and within each group
  $\calT_i$, $i= 3k+1, \dots, s$, orders the voters arbitrarily.
\end{proof}

\newtheorem*{propscfixedeasy}{Proposition~\ref{prop:sc-fixedeasy}}
\begin{propscfixedeasy}
  Given a profile $\calR$ over candidate set $C$, a positive
  integer $k$, and an order $\lhd$, we can decide in polynomial time if there exists a
  profile $\calR'$ with $c(\calR') \geq k$ that is single-crossing with respect to $\lhd$
  and can be obtained from $\calR$ by decloning.
\end{propscfixedeasy}
\begin{proof}
We can assume without loss of generality that $\calR=(R_1, \dots, R_n)$
is single-crossing with respect to the standard order $1 \lhd \cdots \lhd n$ over $[n]$.
Consider any pair of candidates $(c_x, c_y)$ that violates the single-crossing
condition with respect to $\lhd$, i.e., $c_x\succ_1 c_y$, $c_x\succ_n c_y$
and $c_y\succ_i c_x$ for some $i$, $1<i<n$; we will write $x\perp y$ if this is the case.
 Clearly, if $\calR$
is single-crossing, $c_x$ and $c_y$ have to be decloned into the same candidate.
Let $C(x, y)$ be the unique minimal (with respect to set inclusion) 
subset of candidates such that: 
(a) $c_x\in C(x, y)$;
(b) $c_y\in C(x, y)$; and 
(c) for every pair of candidates $c_t, c_z\in C(x, y)$, 
and every candidate $c_w\in C$ such that $c_t\succ_i c_w \succ_i c_z$ for some $i\in[n]$ 
it holds that $c_w\in C(x, y)$.
The set $C(x, y)$ is well-defined and can be constructed inductively: we start with $c_x$ and $c_y$,
and at each step we add all candidates that appear between some candidates
already in $C(x, y)$ in at least one preference order. 
A simple inductive argument shows that $C(x, y)$ is a clone set, 
and, for $\calR'$ to be single-crossing, 
all elements of $C(x, y)$ must be decloned into the same candidate.
We further observe that the set family $\{C(x, y)\mid x\perp y\}$
is laminar, i.e., for every $x, y, z, t\in [m]$ such that $x\perp y$ and $z\perp t$ we have
$C(x, y)\subset C(z, t)$ or $C(z, t)\subset C(x, y)$ or $C(x, y) = C(z, t)$.
Thus, an optimal decloning of $\calR$ can be obtained by decloning
the maximal (with respect to set inclusion) sets in $\{C(x, y)\mid x\perp y\}$.
\end{proof}

\end{document}